\documentclass[11pt]{article}
\usepackage{amssymb}
\usepackage{amsmath}
\usepackage{amsthm}
\usepackage{latexsym}
\usepackage{mathdots}
\usepackage{graphicx}
\usepackage[all]{xy}
\newtheorem{theo}{Theorem}[section]
\newtheorem{prop}[theo]{Proposition}
\newtheorem{cor}[theo]{Corollary}

\theoremstyle{definition}
\newtheorem{defi}[theo]{Definition}
\newtheorem{exa}[theo]{Example}
\newtheorem{rem}[theo]{Remark}
\numberwithin{equation}{section}

\newcommand{\N}{{\mathbb N}}
\newcommand{\F}{{\mathbb F}}

\newcommand{\Z}{{\mathbb Z}}
\newcommand{\Q}{{\mathbb Q}}
\newcommand{\C}{{\mathbb C}}

\newcommand{\cC}{{\mathcal C}}
\newcommand{\cG}{{\mathcal G}}
\newcommand{\cM}{{\mathcal M}}
\newcommand{\cP}{{\mathcal P}}

\newcommand{\wcP}{\widehat{\mathcal P}}
\newcommand{\wwcP}{\widehat{\phantom{\big|}\hspace*{.5em}}\hspace*{-.9em}\wcP}
\newcommand{\cPnsym}{{\mathcal P}_{\rm{sym}}^n}

\newcommand{\cQ}{{\mathcal Q}}

\newcommand{\cL}{{\mathcal L}}

\newcommand{\cR}{{\mathcal R}}
\newcommand{\cS}{{\mathcal S}}
\newcommand{\cN}{{\mathcal N}}

\newcommand{\bP}{{\mathbf P}}
\newcommand{\bH}{{\mathbf H}}
\newcommand{\bbP}{\mathbf{\bar{P}}}

\newcommand{\veps}{{\varepsilon}}
\newcommand{\widesim}[1][1.5]{\scalebox{#1}[1]{$\sim$}}

\newcommand{\comp}{\text{comp}}

\newcommand{\wt}{{\rm wt}}

\newcommand{\pe}{\mbox{\rm PE}}

\newcommand{\mmid}{\mbox{$\,|\,$}}
\newcommand{\mmidbig}{\mbox{$\,\big|\,$}}

\newcommand{\T}{\mbox{$\!^{\sf T}$}}

\newcommand{\inner}[1]{\mbox{$\langle{#1}\rangle$}}
\newcommand{\ideal}[1]{\mbox{$\langle{#1}]$}}

\newcommand{\supp}{\mbox{\rm supp}}

\newcommand{\pperp}{\perp \! \! \! \perp}

\newcounter{alp}
\newcounter{ara}
\newcounter{rom}

\newenvironment{alphalist}{\begin{list}{(\alph{alp})\hfill}{\usecounter{alp}
     \topsep0ex \labelwidth.6cm \leftmargin.6cm \labelsep0cm
     \rightmargin0cm \parsep0ex \itemsep0ex
     \partopsep1.6ex}}{\end{list}}

\title{Fourier-Reflexive Partitions and MacWilliams Identities\\ for Additive Codes}
\date\today
\author{Heide Gluesing-Luerssen\thanks{The author was partially supported by the National Science Foundation
        grants \#DMS-0908379 and \#DMS-1210061.}\\
        University of Kentucky\\ Department of Mathematics\\
       715 Patterson Office Tower\\ Lexington, KY 40506-0027, USA\\ heide.gl@uky.edu}

\begin{document}
\maketitle

{\bf Abstract:}
A partition of a finite abelian group gives rise to a dual partition on the character group via the
Fourier transform.
Properties of the dual partitions are investigated and a convenient test is given for the case that the
bidual partition coincides the primal partition.
Such partitions permit MacWilliams identities for the partition enumerators of additive codes.
It is shown that dualization commutes with  taking products and symmetrized products of
partitions on cartesian powers of the given group.
After translating the results to Frobenius rings, which are identified with their character module, the approach is
applied to partitions that arise from poset structures.

%

\section{Introduction}\label{SS-Intro}
\setcounter{equation}{0}

MacWilliams identities relate properties of a code, in form of an enumerator or distribution, to properties of the dual code, and this relation is made precise by a concrete transformation.
Such identities form an inevitable tool for the theory of self-dual codes and also help to derive linear programming bounds for codes, see, e.g.,~\cite{Del73,BGO07}.
More engineering-oriented interest in MacWilliams identities can be found in~\cite{KhMcE05,LKY04}.

The most famous MacWilliams identity is the one for the Hamming weight enumerator for codes over fields derived by MacWilliams in~\cite{MacW62}.
It has been generalized to other weight functions, most notably the symmetrized Lee weight, by MacWilliams and Sloane in
the monograph~\cite{MS77}.

Starting in the 80's, these results have been further generalized to codes over additive groups and over finite (commutative and later non-commutative) Frobenius rings by Delsarte~\cite{Del73},  Klemm~\cite{Kl87,Kl89}, Nechaev and Kuzmin~\cite{Nec96,NeKu97}, and  Wood~\cite{Wo99}.
The area of codes over rings gained even more attention after discovering the relevance of the Lee weight on~$\Z_4$ for understanding the
formal duality of the binary non-linear Kerdock and Preparata codes by Hammons et al.~\cite{HKCSS94}.
In the realm of MacWilliams identities, the Frobenius property of the ring is essential as it guarantees that the character-theoretic
annihilator of a code can be identified with the standard dual.

The most general approach has been taken by Delsarte~\cite{Del73} with the aid of association schemes; see also
Camion~\cite{Cam98} and Delsarte and Levenshtein~\cite{DeLe98}.
This approach allows to also deal with non-linear codes and their distance distributions (again, with respect to various distance functions).
Unfortunately, Delsarte's results have not gained the attention they deserve, which may be due to the extended
machinery of association schemes and their related Bose-Mesner algebras.
Many MacWilliams identities that form just special cases of Delsarte's approach, have been proven independently afterwards.

In the middle of the 90's, Zinoviev and Ericson~\cite{ZiEr96} revived the fundamental ideas of Delsarte by studying partitions of the ambient space, thus classifying the words according to a pre-specified property (e.g. the Hamming weight).
In~\cite{ZiEr96}, they focus on additive codes (i.e., codes over additive groups) and study the case where the
enumerators of the code and its dual with respect to the same partition obey a MacWilliams identity.
This leads to the notion of an F-partition or Fourier partition, reminding of the fact that the space of indicator functions of the partition sets is invariant under the Fourier transform.
In~\cite{ZiEr09} the same authors extend their ideas to admissible pairs of partitions; these are partition pairs $(\cP,\cQ)$ for which the Fourier transform induces an isomorphism between the two associated spaces of indicator functions.
For such pairs the $\cP$-enumerator of a code and the $\cQ$-enumerator of the dual code obey a MacWilliams identity.
These pairs are exactly those inducing an abelian association scheme on the second cartesian power of the group, thus
relating this approach to Delsarte's in~\cite{Del73}.
In~\cite{Fo98} Forney presents a generalization of this result to discrete subgroups of locally compact abelian groups.

In~\cite{HoLa01} Honold and Landjev use a similar approach for linear codes over finite Frobenius rings.
By not using characters and the Fourier transform but rather a certain map with a particular homogeneous property, their partition pairs are characterized in a different way.

In this paper we will present an approach to MacWilliams identities solely based on partitions of a group and their duals.
This will also allow us to survey the vast literature on the subject in more detail.
We will make a careful distinction between the ambient space of the code and that of its dual: the code will be a subgroup of a given
finite abelian group, and its dual will reside in the character group.
This is different from many of the settings above, where the code and its dual are contained in the same ambient space, mostly by identifying
the group with its character group.
However, since this identification is not canonical, many partition properties, most notably the dual of a partition, depend on the choice of the identification.
Fortunately, this undesirable behavior does not occur for many standard situations, e.g., the Hamming weight on any group, the complete weight, or for any partition on~$\Z_N$.

By not identifying a group with its character group, one is imperatively in the situation that the enumerators of a code and its
dual refer to different partitions, as considered by Zinoviev and Ericson in~\cite{ZiEr09}.
Closely related is a notion introduced by Byrne et al.\  in~\cite{BGO07}.
In the setting of codes over a finite Frobenius ring, identified with its character-module, they introduce the dual of a given partition with respect to the Fourier transform.
This can be regarded as the one-sided analogue of the admissible pairs in~\cite{ZiEr09} as now the two partitions need not be mutually dual.
One of their motivation for dualizing a given partition appears to be the homogeneous weight.
This weight function induces in general a partition that is not an F-partition.

In this paper we will draw ideas from many of the above mentioned papers and combine them to a unified theory of partitions, partition enumerators and MacWilliams identities for additive codes.
Our starting point will be the dualization of a partition as introduced by Byrne et al.~\cite{BGO07}.
We will characterize when the bidual of a partition coincides with the given partition, and will call such partitions
(Fourier-) reflexive.
They correspond to the admissible pairs in~\cite{ZiEr09} and to abelian association schemes investigated
in~\cite{Del73}.
We will derive basic properties of dual partitions and present the resulting MacWilliams identity for the associated partition enumerators.
Particular emphasis will be placed on product partitions and symmetrized product partitions of a cartesian
power of a finite abelian group.
We will show that these constructions commute with dualization.

We will derive all results directly from our definitions in order to show how they fall naturally (and easily) into place.
Whenever applicable, we will refer to analogous or closely related results in the literature; they are often for a slightly different
or restricted setting.

Next, we will briefly translate the setting to finite commutative Frobenius rings, in which case we will identify the ring with its
character module.

In Section~\ref{SS-Posets} we will apply our approach to a particular recent area of MacWilliams identities, where the weight function is
based on poset structures on the underlying index set $\{1,\ldots,n\}$.
This has been investigated by Kim and Oh~\cite{KiOh05} and Pinheiro and Firer~\cite{PiFi12}.
The resulting weight generalizes the Rosenbloom-Tsfasman weight introduced in~\cite{RoTs97}.
By studying the underlying partitions we will generalize the MacWilliams identities to codes over groups.
Once the desired partition duality is established, the identities are simply examples of the general theory.

In a separate paper~\cite{GL13homog}, we use our approach to investigate for which Frobenius rings the homogeneous weight
induces a reflexive partition.



\section{Partitions and Duality on Finite Abelian Groups}\label{SS-Prep}
\setcounter{equation}{0}
In this section we review the basic material on character theory and the Fourier transform.
We introduce the dual of a partition and present the resulting MacWilliams identity.

Throughout, let~$G$ be a finite abelian (additive) group.
Its \emph{character group} is defined as $\hat{G}:=\text{Hom}(G,\C^*)$ along with addition
$(\chi_1+\chi_2)(g):=\chi_1(g)\chi_2(g)$ for all~$\chi_i\in\hat{G}$.
We prefer to write the operation additively because if~$G$ is the additive group of ring, then the character group has a
module structure with exactly this addition (see Section~\ref{SS-FrobRing}).
Due to the finiteness of~$G$, each character value~$\chi(g)$ is a root of unity in~$\C$.
As a consequence, the negative of the character~$\chi$ is given by $(-\chi)(g):=\chi(g)^{-1}=\overline{\chi(g)}$, where
$\overline{\phantom{T}}$ denotes complex conjugation.
The zero element of~$\hat{G}$ is the trivial map $\chi\equiv1$.
It is called the \emph{principal character} of~$G$ and is denoted by~$\veps$.
It is well-known~\cite{Te99}, that the groups~$G$ and~$\hat{G}$ are (non-canonically) isomorphic, and thus $|G|=|\hat{G}|$.
The groups~$\hat{\hat{G}}$ and~$G$ are canonically  isomorphic by
mapping~$g\in G$  to the character of~$\hat{G}$ that sends $\chi\in\hat{G}$ to~$\chi(g)$.
Thus $g(\chi)=\chi(g)$.
This suggests to write $\inner{\chi,g}:=\chi(g)$, and thus we have the identities
$\inner{\chi,g}=\inner{g,\chi}$ as well as
$\inner{\chi,g_1+g_2}=\inner{\chi,g_1}\inner{\chi,g_2}$ and $\inner{\chi_1+\chi_2,g}=\inner{\chi_1,g}\inner{\chi_2,g}$.

It is easy to see~\cite[p.~171]{Te99} that
$(G_1\times\ldots\times G_n)\!\widehat{\phantom{I}}=\hat{G_1}\times\ldots\times \hat{G_n}$ for any groups~$G_1,\ldots,G_n$, and where the character maps are given by
\begin{equation}\label{e-prodchar}
   \inner{(\chi_1,\ldots,\chi_n),(g_1,\ldots,g_n)}:=\prod_{i=1}^n \inner{\chi_i, g_i} \text{ for all }\chi_i\in \hat{G_i},\, g_i\in G_i.
\end{equation}

%

A subgroup of~$G$ is called an \emph{(additive) code over} $G$.
For a code $\cC\leq G$, the \emph{dual code} $\cC^\perp\leq\hat{G}$ is defined as
\begin{equation}\label{e-DualSubgroup}
   \cC^\perp=\{\chi\in\hat{G}\mid \inner{\chi,h}=1\text{ for all }h\in \cC\}.
\end{equation}
It is straightforward to see that $\cC^\perp\cong\widehat{G/\cC}$, and thus
$|\cC^{\perp}|\cdot|\cC|=|G|$.
As a consequence,
\begin{equation}\label{e-bidual}
    \cC^{\perp\perp}=\cC\text{ for any code }\cC\leq G.
\end{equation}

The most important tool for deriving MacWilliams identities are the orthogonality relations~\cite[Lem.~(1.1.32)]{vLi99}
\begin{equation}\label{e-GChar}
   \sum_{g\in G}\inner{\chi,g}=\left\{\begin{array}{cl}
          0,&\text{if }\chi\neq\veps,\\ |G|,&\text{if }\chi=\veps.\end{array}\right.
\end{equation}
For a code $\cC\leq G$ they lead immediately to
\begin{equation}\label{e-HChar}
  \sum_{h\in \cC}\inner{\chi,h}=0\text{ if } \chi\not\in \cC^\perp\ \text{ and }\
  \sum_{h\in \cC}\inner{\chi,h}=|\cC|\text{ if } \chi\in \cC^\perp.
\end{equation}

A starting point for proving MacWilliams identities is the Poisson summation formula for maps on~$G$ and their
Fourier transforms.
Let~$V$ be any complex vector space and  $f:\, G\longrightarrow V$ be any map.
The \emph{Fourier transform} of~$f$ is defined as~\cite[p.~261]{Te99}
\begin{equation}\label{e-Fourier}
  f^+:\, \hat{G}\longrightarrow V,\quad \chi\longmapsto \sum_{g\in G}\inner{\chi,g}f(g).
\end{equation}
The Fourier transform is invertible, and with the standard identification of~$G$ with $\hat{\hat{G}}$ and~\eqref{e-GChar} one easily verifies
\begin{equation}\label{e-FouInv}
  f^{++}(g)=|G|f(-g).
\end{equation}
A map~$f$ and its Fourier transform satisfy the \emph{Poisson summation formula}~\cite[p.~199]{Te99}
\begin{equation}\label{e-Poisson}
    \sum_{\chi\in\cC^{\perp}}f^+(h)=|\cC^{\perp}|\sum_{h\in\cC}f(h)
\end{equation}
for any code $\cC\leq G$.

\medskip
We now turn to partitions on~$G$ and their dual partitions on~$\hat{G}$.
Let us first fix the following notation.
A partition $\cP=(P_m)_{m=1}^M$ of a set~$X$, i.e., the sets~$P_m$ are disjoint and cover~$X$, will mostly be
written as $\cP=P_1\mmid P_2\mmid\ldots\mmid P_M$.
The sets of a given partition are called its \emph{blocks}.
We write $|\cP|$ for the number of blocks in~$\cP$.
Two partitions~$\cP$ and~$\cQ$ of a set~$X$ are called \emph{identical} if $|\cP|=|\cQ|$ and the blocks coincide after suitable indexing.
Moreover,~$\cP$ is called \emph{finer} than~$\cQ$ (or~$\cQ$ is \emph{coarser} than~$\cP$), written as $\cP\leq\cQ$, if
for every block~$P$ of~$\cP$ there exists a block~$Q$ of~$\cQ$ such that $P\subseteq Q$.
Note that if~$\cP\leq\cQ$ then $|\cP|\geq|\cQ|$.
Denote by~$\widesim_{\cP}$ the equivalence relation induced by~$\cP$, thus,
$v\widesim_{\cP}v'$ if $v,\,v'$ are in the same block of~$\cP$.

The following notion of a dual partition will be central to our presentation.
It has been introduced by Byrne et~al.~\cite[p.~291]{BGO07} and forms an excellent setting to encompass
various situations discussed in the literature.
It goes back to the notion of F-partitions as introduced by Zinoviev and Ericson in~\cite{ZiEr96}.
The special case of (Fourier-) reflexive partitions, defined below, corresponds to abelian association schemes as
studied by Delsarte~\cite{Del73}, Camion~\cite{Cam98}, and others; see also~\cite{DeLe98} and~\cite{ZiEr09}.
If we identify the group~$G$ with its character group~$\hat{G}$, reflexive partitions turn out to be exactly
the B-partitions studied by Zinoviev and Ericson~\cite{ZiEr09}.
In this case one may also ask whether a partition is self-dual, i.e., coincides with its dual.
However, the dual partition, and thus self-duality, in general depends on the choice of the identifying isomorphism
between~$G$ and~$\hat{G}$.
Since we will do identify~$G$ with~$\hat{G}$ we will not embark on this direction.
Only in Section~\ref{SS-FrobRing}, when considering finite commutative Frobenius rings we will touch upon this issue.

\begin{defi}\label{D-FPart}
Let $\cP=P_1\mmid P_2\mmid\ldots\mmid P_M$ be a partition of~$G$.
The \emph{dual partition}, denoted by~$\wcP$, is the partition of~$\hat{G}$ defined via the equivalence relation
\begin{equation}\label{e-simPhat}
  \chi\widesim_{\wcP} \chi' :\Longleftrightarrow \sum_{g\in P_m}\inner{\chi,g}=\sum_{g\in P_m}\inner{\chi',g} \text{ for all }m=1,\ldots,M.
\end{equation}
Let $\wcP=Q_1\mmid\ldots\mmid Q_L$.
The \emph{generalized Krawtchouk coefficients}~$K_{\ell,m}$ are defined as
\begin{equation}\label{e-klm}
    K_{\ell,m}=\sum_{g\in P_m}\inner{\chi,g},\text{ where $\chi$ is any element in }Q_l.
\end{equation}
The matrix $K=(K_{\ell,m})\in\C^{L\times M}$ is the \emph{generalized Krawtchouk matrix} of $(\cP,\wcP)$.
Finally, the partition~$\cP$ is called \emph{Fourier-reflexive} or simply \emph{reflexive} if $\wwcP=\cP$.
\end{defi}
Note that by definition of the dual partition, the Krawtchouk coefficient $K_{\ell,m}$ does not depend on the choice
of~$\chi$ in~$Q_{\ell}$.
The relation to the Fourier transform will be detailed below.

The following properties are easy to verify.

\begin{rem}\label{R-DualNega}
Let $\cP=P_1\mid \ldots\mid P_M$.
\begin{alphalist}
\item The singleton~$\{\veps\}$ is always a block of~$\wcP$.
       This follows immediately from
       $\sum_{m=1}^M\sum_{g\in P_m}\inner{\veps,g}=\sum_{g\in G}\inner{\veps,g}$ along
       with the orthogonality relations~\eqref{e-GChar}.
       In particular, if~$\cP$ consists of the single block~$G$, then $\wcP=\{\veps\}\mmid \hat{G}\backslash\{\veps\}$,
       which again is a simple consequence of~\eqref{e-GChar}.
\item $\wcP=-\wcP=\widehat{-\cP}$, where $-\cP:=-P_1\,\mid -P_2\,\mid\ldots\mid -P_M$ and
         $-P:=\{-g\mid g\in P\}$ for a set~$P$.
        This follows from
      $\sum_{g\in P_m}\inner{\chi,-g}=\sum_{g\in P_m}\inner{-\chi,g}=\sum_{g\in P_m}\!\inner{\chi,g}^{-1}
       =\overline{\sum_{g\in P_m}\phantom{\big|}\hspace*{-.4em}\inner{\chi,g}}$ (the complex conjugate).
\item If $\cP\leq\cQ$, then $\wcP\leq \widehat{\cQ}$.
      This is clear with~\eqref{e-simPhat} and the fact that each block of~$\cQ$ is a union of blocks of~$\cP$.
\end{alphalist}
\end{rem}

\begin{exa}\label{E-Z6Z8}
\begin{alphalist}
\item Consider $G=\Z_N$ for some $N>1$. Then the character group is given by $\{\chi_a\mid a\in\Z_N\}$, where
        $\inner{\chi_a,g}=\zeta^{ag}$ with a fixed primitive $N$-th root of unity $\zeta\in\C$.
        In this way, $a\mapsto\chi_a$ furnishes an isomorphism of~$G$ and~$\hat{G}$.
        Consider now $\Z_6$ and identify the group with its character group in the above way.
       Let $\cP=0\mmid1,3,5\mmid2,4$ (where we omit the parentheses).
      Then one computes $\widehat{\cP}=0\mmid1,2,4,5\mmid3$.
      To be precise, with a $6$-th primitive root of unity
      $\sum_{g\in P_2}\inner{\chi_a,g}=\zeta^{a\cdot1}+\zeta^{a\cdot3}+\zeta^{a\cdot5}$ has value~$0$ for $a=1,2,4,5$ and
      value~$-1$ for $a=3$ and similarly for the sum over the block $\{2,4\}$.
      Computing all other sums we obtain the Krawtchouk matrix
      \[
          \begin{pmatrix}1&3&2\\1&0&-1\\1&-3&2\end{pmatrix},
      \]
      where the rows and columns are indexed by the blocks of~$\widehat{\cP}$ and~$\cP$ in the given order.
      In the same way one can compute~$\wwcP$ and obtains $\wwcP=\cP$. Hence~$\cP$ is reflexive.
\item For the partition $\cP=0\mid 1,2\mid 3,4,5$ of~$\Z_6$ one computes $\wcP=0\mid 1\mid 2,4\mid 3\mid 5$, while
      $\wwcP$ consists of the singletons. Thus, $\cP$ is not reflexive.
\item This example will be needed later in Section~\ref{SS-Posets}. 
      Let~$G=A_1\times\ldots\times A_n$, where~$A_i$ is a finite abelian group of order~$q\geq2$ for all~$i$.
      Let~$\cP$ be the partition of~$G$ given by the Hamming weight, that is,
      $P_m=\{a\in G\mid \wt(a)=m\}$ for $m=0,\ldots,n$, and where $\wt(a_1,\ldots,a_n)=|\{i\mid a_i\neq0\}|$.
      Analogously, let~$\cQ$ be the Hamming partition of~$\hat{G}=\hat{A_1}\times\ldots\times\hat{A_n}$, thus
      $Q_\ell=\{\chi\in\hat{G}\mid \wt(\chi)=\ell\}$.
      Note that $\wt(\chi_1,\ldots,\chi_n)=|\{i\mid \chi_i\neq\veps\}|$ since~$\veps$ is the zero element of~$\hat{G}$.
      Then it is well known, see for instance \cite[Thm.~4.1]{Del73} or Lemma~2.6.2 in~\cite{HP03}, that
      $\sum_{a\in P_m}\inner{\chi,a}=K_m^{(n,q)}(\ell)$ for each $\chi\in Q_\ell$, and
      where
      \begin{equation}\label{e-Krawtclassic}
          K_m^{(n,q)}(x)=\sum_{j=0}^m(-1)^j(q-1)^{m-j}\binom{x}{j}\binom{n-x}{m-j}
       \end{equation}
       is the Krawtchouk polynomial
       (the proof of Lemma~2.6.2 in~\cite{HP03}, given for $\Z_q^n$, works mutatis mutandis for all
       $A_1\times\ldots\times A_n$).
      Since $K_1^{(n,q)}(\ell)\neq K_1^{(n,q)}(\ell')$ whenever $\ell\neq\ell'$, this shows that~$\cQ$ is, as expected, the
      dual partition of~$\cP$, and thus $K_{\ell,m}=K_m^{(n,q)}(\ell)$ are the Krawtchouk coefficients of~$(\cP,\cQ)$.
      By symmetry,~$\cP$ is reflexive.
      Identifying the isomorphic groups~$G$ and~$\hat{G}$ we see that we may call~$\cP$ self-dual, i.e.,
      $\cP=\wcP$ (with respect to any isomorphism between~$G$ and~$\hat{G}$).
\end{alphalist}
\end{exa}

The rest of this section is devoted to a MacWilliams identity for partition enumerators of codes in~$G$ and their dual codes
in~$\hat{G}$.
The result can also be found in~\cite[Thm.~4.72, Prop.~5.42]{Cam98} by Camion, where it has been derived with
the aid of association schemes, and in~\cite[p.~94]{Fo98} by Forney for discrete subgroups of locally compact abelian groups.
In~\cite[p.~88]{Del73} Delsarte notices already the connection between abelian association schemes and MacWilliams
identities.
In the form of~\eqref{e-MacWPart} below and for the special case of self-dual partitions (i.~e.,~$\cP=\wcP$ when
identifying $G$ and~$\hat{G}$), the identity is established by Zinoviev/Ericson~\cite[Thm.~1]{ZiEr96}.
For the case of submodules of~$R^n$, where~$R$ is a Frobenius ring identified with~$\hat{R}$, and where
the partition~$\cQ$ arises from symmetrization of a partition on~$R$, Theorem~\ref{T-MacWPart} below has
been established by  Byrne et al.~\cite[Thm.~2.11]{BGO07}.

Let $\cP=P_1\mmid\ldots\mmid P_M$  and $\cQ=Q_1\mmid\ldots\mmid Q_L$ be partitions of~$G$ and~$\hat{G}$, respectively,
such that $\cP=\widehat{\cQ}$ (note that~$\cQ$ is the ``initial'' partition and~$\cP$ its dual).
Let $K=(K_{m,\ell})\in\C^{M\times L}$ be the Krawtchouk matrix of~$(\cQ,\cP)$.
For a code $\cC\leq G$ and its dual $\cC^\perp\leq\hat{G}$  define the partition enumerators
$ \pe_{\cP,\cC}\in\C[X_1,\ldots,X_M]$ and $\pe_{\cQ,\cC^{\perp}}\in\C[Y_1,\ldots,Y_L]$ as
\begin{equation}\label{e-peboth}
    \pe_{\cP,\cC}=\sum_{m=1}^M A_m X_m,\quad
    \pe_{\cQ,\cC^{\perp}}=\sum_{\ell=1}^L B_\ell Y_\ell,
    \text{ where }A_m=|\cC\cap P_m|\text{ and } B_\ell=|\cC^\perp\cap Q_\ell|.
\end{equation}
They carry the information about the number of codewords contained in each block of the partitions.

\medskip
The following form of MacWilliams identity provides a transformation~$\cM$ of the enumerator~$\pe_{\cP,\cC}$
resulting in the enumerator $\pe_{\cQ,\cC^{\perp}}$.
Its well-definedness is guaranteed by the relation $\cP=\widehat{\cQ}$.
In general, it is not possible to invert the operator and compute $\pe_{\cP,\cC}$ from $\pe_{\cQ,\cC^{\perp}}$.
This is only guaranteed if $\cQ=\wcP$, i.e., if~$\cP$ and~$\cQ$ are both reflexive and thus mutually dual.
Thus, not surprisingly, reflexive partitions provide a symmetric situation and form the most appealing case.

\begin{theo}\label{T-MacWPart}
Define the MacWilliams transformation $\cM:\,\C[X_1,\ldots,X_M]\longrightarrow\C[Y_1,\ldots,Y_L]$  as the algebra
homomorphism given by $\cM(X_m)=\sum_{\ell=1}^L K_{m,\ell}Y_\ell$ for $m=1,\ldots,M$.
Then
\begin{equation}\label{e-MacWPart}
  \pe_{\cQ,\cC^{\perp}}=\frac{1}{|\cC|}\cM(\pe_{\cP,\cC}).
\end{equation}
\end{theo}
\begin{proof}
Define the maps $\psi:G\rightarrow \C[X_1,\ldots,X_M],\ g\mapsto X_{m(g)}$, where~$m(g)$ is the unique index
such that $g\in P_{m(g)}$.
Similarly, let $\bar{\psi}:\hat{G}\rightarrow \C[Y_1,\ldots,Y_L],\ \chi\mapsto Y_{\ell(\chi)}$, where~$\ell(\chi)$ is such that
$\chi\in Q_{\ell(\chi)}$.
Then $\pe_{\cP,\cC}=\sum_{g\in\cC}X_{m(g)}$ and $\pe_{\cQ,\cC^{\perp}}=\sum_{\chi\in\cC^{\perp}}Y_{\ell(\chi)}$.
Using the definition of the Krawtchouk coefficients  we obtain for $g\in P_m$
\[
  \cM(X_m)=\sum_{\ell=1}^LK_{m,\ell}Y_\ell=\sum_{\ell=1}^L\sum_{\chi\in Q_\ell}\inner{\chi,g}\bar{\psi}(\chi)
  =\sum_{\chi\in\hat{G}}\inner{\chi,g}\bar{\psi}(\chi)=\bar{\psi}^+(g),
\]
where~$\bar{\psi}^+$ is the Fourier transform.
The Poisson formula~\eqref{e-Poisson} applied to the map~$\bar{\psi}$ yields
\[
   \cM(\pe_{\cP,\,\cC})=\sum_{g\in\cC}\cM(X_{m(g)})=\sum_{g\in\cC}\bar{\psi}^+(g)=|\cC|\sum_{\chi\in\cC^{\perp}}\bar{\psi}(\chi)
   =|\cC|\pe_{\cQ,\,\cC^{\perp}}.
   \qedhere
\]
\end{proof}

One should also observe that~\eqref{e-MacWPart} is simply the linear identity
\begin{equation}\label{e-MacWLinear}
   (B_1,\ldots,B_L)=\frac{1}{|\cC|}(A_1,\ldots,A_M)K,
\end{equation}
where~$K$ is the Krawtchouk matrix.

It is easy to see that the MacWilliams identity can be generalized straightforwardly to the following situation.
Suppose $\cP'\leq\cP$ and $\cQ'\geq\cQ$, i.e., $\cP'$ is finer than~$\cP$ and~$\cQ'$ is coarser than~$\cQ$.
Then one obtains a MacWilliams transformation of~$\pe_{\cP',\cC}$ resulting in $\pe_{\cQ',\cC^{\perp}}$.
For symmetrized partitions on~$R^n$, where~$R$ is a Frobenius ring, this has been  presented by Byrne et
al.~\cite[Thm.~2.11]{BGO07}.

\section{Reflexive Partitions}\label{SS-ReflPart}
In this section we preset a characterization of reflexivity and discuss some further properties of dual partitions with an
emphasis on reflexive partitions.

We need to describe the dualization of partitions in terms of the Fourier transform.
Similar considerations can be found in the papers~\cite{ZiEr96,ZiEr09} by Zinoviev and Ericson.
Let
\[
    \cP=P_1\mmid P_2\mmid\ldots\mmid P_M\ \text{ and }\ \wcP=Q_1\mmid Q_2\mmid\ldots\mmid Q_L
\]
be partitions of~$G$ and~$\hat{G}$.
Denote by~$\psi_m$ the \emph{indicator function} of the block~$P_m,\,m=1,\ldots,M$.
In the vector space~$\C^{G}$ of maps from~$G$ to~$\C$, consider the $m$-dimensional  subspace
$\cL(\cP)=\inner{\psi_1,\ldots,\psi_M}_{\C}$ generated by the functions~$\psi_m$.
The Fourier transforms of the indicator functions are given by $\psi_m^+(\chi)=\sum_{g\in P_m}\inner{\chi,g}$,
and thus we may write for $g,\,g'\in G$ and $\chi,\,\chi'\in\hat{G}$
\begin{align}
  g\widesim_{\cP}g'&\Longleftrightarrow \psi_m(g)=\psi_m(g') \text{ for all }m=1,\ldots,M, \label{e-FourEquiva}\\
  \chi\widesim_{\wcP}\chi'&\Longleftrightarrow \psi_m^+(\chi)=\psi_m^+(\chi') \text{ for all }m=1,\ldots,M. \label{e-FourEquivb}
\end{align}
The last equivalence shows that the functions~$\psi_m^+$ are constant on each block~$Q_\ell$ of~$\wcP$.
In other words,
\begin{equation}\label{e-L+}
   \cL^+(\cP):=\inner{\psi_1^+,\ldots,\psi_M^+}_{\C}\subseteq\cL(\wcP):=\inner{\xi_1,\ldots,\xi_L}_{\C}, 
\end{equation}
where~$\xi_\ell$ denotes the indicator function of the dual block~$Q_\ell$.
Now we are in a position to characterize reflexivity.
It leads to the convenient criterion that if the dual partition~$\wcP$ has the same number of
blocks as~$\cP$, then~$\cP$ is reflexive; in particular, taking further duals does not increase the number of blocks
of the partitions.
In the language of association schemes, the following criterion for reflexivity can also be found in~\cite[Fact.~V.2]{Hon10}.

\begin{theo}\label{P-LPplus}
We have $|\cP|\leq|\wcP|$ and $\wwcP\leq\cP$.
Moreover, $\cP$ is reflexive if and only if $|\cP|=|\wcP|$.
\end{theo}

\begin{proof}
Injectivity of the Fourier transform and~\eqref{e-L+} yield $|\cP|=M=\dim\cL^+(\cP)\leq\dim\cL(\wcP)=L=|\wcP|$.
This proves the first statement.
\\
Next, the inverse Fourier transform in~\eqref{e-FouInv}, Remark~\ref{R-DualNega}(b), and~\eqref{e-L+} yield
    $\cL(-\cP)=\cL^{++}(\cP)\subseteq\cL^+(\wcP)=\cL^+(-\wcP)$.
    This in turn implies
    $\cL(\cP)\subseteq\cL^+(\wcP)=\inner{\xi_0^+,\ldots,\xi_L^+}_{\C}$.
    In other words, the indicator functions $\psi_m$ are linear combinations of~$\xi_0^+,\ldots,\xi_L^+$.
    Now~\eqref{e-FourEquivb} for the dual partition shows that if
    $g\widesim_{\mbox{\scriptsize$\wwcP$}\,}g'$ then $\xi_\ell^+(g)=\xi_\ell^+(g')$ for all~$\ell$, thus
    $\psi_m(g)=\psi_m(g')$ for all~$m$.
    With~\eqref{e-FourEquiva} we conclude $g\widesim_{\cP}g'$.
    All of this shows $\wwcP\leq\cP$.
\\
It remains to prove the characterization of reflexivity.
The only-if part follows from the first part.
For the converse it remains to show $\cP\leq\wwcP$.
By~\eqref{e-L+} we have $\cL^+(\cP)=\cL(\wcP)$.
The same reasoning as above implies
$\cL(\cP)=\cL^+(\wcP)=\inner{\xi_0^+,\ldots,\xi_L^+}_{\C}$, and this means
that each~$\xi_\ell^+$ is a linear combination of $\psi_1,\ldots,\psi_M$.
Let now $g,\,g'\in G$ be such that $g\widesim_{\cP}g'$.
Then $\psi_m(g)=\psi_m(g')$ for all~$m$, and
hence $\xi_\ell^+(g)=\xi_\ell^+(g')$ for all~$\ell$.
Thus $g\widesim_{\mbox{\scriptsize$\wwcP$}\,}g'$,  and this establishes $\cP\leq\wwcP$.
\end{proof}

As already hinted at, reflexive partitions are closely related to abelian association schemes.
Indeed, with the above result and some straightforward, but lengthy computations one can show that the partition~$\cP$
is reflexive if and only the partition~$\cR=R_1\mmid R_2\mmid\ldots$ $\mmid R_M$ of $G\times G$ defined via
$(x,y)\in R_m\Leftrightarrow x-y\in P_m$ is an abelian association scheme.
With the machinery of association schemes, this has already been established in \cite[Cor.~4.51]{Cam98} as well
as~\cite[Thm.~1]{ZiEr09} and goes back to \cite[Sec.~2.6.1]{Del73}.

\medskip
We now return to the general situation of dualizing partitions.
It is natural to ask whether the operations on the lattice of partitions of~$G$ (w.r.t.~$\leq$) are respected by
dualization.
Let~$\cP$ and $\cQ$ be partitions. Recall that the \emph{join} $\cP\vee\cQ$ is defined as the finest partition that
is coarser than both~$\cP$ and~$\cQ$, and the \emph{meet} $\cP\wedge\cQ$ is
defined as coarest partition that is finer than both~$\cP$ and~$\cQ$.
The following examples show that in general $\widehat{\cP\vee\cQ}\neq\widehat{\cP}\vee\widehat{\cQ}$ and $\widehat{\cP\wedge\cQ}\neq\widehat{\cP}\wedge\widehat{\cQ}$.
\begin{exa}\label{E-MeetJoin}
\begin{alphalist}
\item On the group~$\Z_8$ (identified with its character group~$\hat{G}$ as in Example~\ref{E-Z6Z8}(a)) consider
      $\cP=0\mmid1,7\mmid2,6\mmid 3,5\mmid4$ and $\cQ=0\mmid1,3\mmid2,6\mmid4\mmid5,7$.
      One easily checks that $\wcP=\cP$ and $\widehat{\cQ}=\cQ$.
      However, $\wcP\wedge\widehat{\cQ}=\cP\wedge\cQ=0\mmid1\mmid 2,6\mmid3\mmid4\mmid5\mmid7$, whereas
      $\widehat{\cP\wedge\cQ}$
      consists of the singletons in~$\Z_8$.
\item On $G=\Z_5$ (identified with its character group) consider the partitions
      $\cP=0\mmid1,2\mmid3,4$ and $\cQ=0\mmid1,2,3\mmid4$.
      Then $\wcP$ and~$\widehat{\cQ}$ both consist of the singletons in~$\Z_5$, and hence so does $\wcP\vee\widehat{\cQ}$.
      On the other hand, $\cP\vee\cQ=0\mmid1,2,3,4=\widehat{\cP\vee\cQ}$.
\end{alphalist}
\end{exa}

However, we have the following result.
It has also been derived with the theory of association schemes in \cite[Sec.~4.13]{Cam98}.

\begin{prop}\label{P-JoinMeet}
Let~$\cP$ and~$\cQ$ be reflexive partitions of~$G$. Then the join $\cP\vee\cQ$ is reflexive and
$\widehat{\cP\vee\cQ}=\widehat{\cP}\vee\widehat{\cQ}$.
Example~\ref{E-MeetJoin}(a) shows that the analogous result for the meet~$\cP\wedge\cQ$ is in general not true.
\end{prop}
\begin{proof}
We have $\wwcP=\cP$ and $\widehat{\phantom{\big|}\hspace*{.5em}}\hspace*{-.9em}\widehat{\cQ}=\cQ$.
Moreover, by definition of the join, $\cP\leq \cP\vee\cQ$, and thus $\wcP\leq\widehat{\cP\vee\cQ}$ due to
Remark~\ref{R-DualNega}(c).
Thus, again by definition of the join, $\wcP\vee\widehat{\cQ}\leq\widehat{\cP\vee\cQ}$.
Taking duals we obtain
$\cP=\widehat{\phantom{\big|}\hspace*{.5em}}\hspace*{-.9em}\wcP\leq
  \widehat{\phantom{\big|}\hspace*{2.3em}}\hspace*{-2.7em}\widehat{\cP\vee\cQ}$ and thus
$\cP\vee\cQ\leq \widehat{\phantom{\big|}\hspace*{2.3em}}\hspace*{-2.7em}\widehat{\cP\vee\cQ}$.
But then Theorem~\ref{P-LPplus} yields
$\cP\vee\cQ=\widehat{\phantom{\big|}\hspace*{2.3em}}\hspace*{-2.7em}\widehat{\cP\vee\cQ}$, which proves
the reflexivity of the join.
Applying this line of reasoning to $\wcP\vee\widehat{\cQ}$, we obtain
$\wcP\vee\widehat{\cQ}=\widehat{\phantom{\Big|}\hspace*{2.3em}}\hspace*{-2.7em}\widehat{\wcP\vee\widehat{\cQ}}
  \geq\widehat{\cP\vee\cQ}$.
All of this shows $\wcP\vee\widehat{\cQ}=\widehat{\cP\vee\cQ}$, as desired.
\end{proof}

The following property of the Krawtchouk matrix is known for reflexive partitions, in which case
it simply states that the matrix is orthogonal, see~\cite[p.~12]{Del73}.

\begin{prop}\label{P-KKtilde}
Let $\cP=P_1\mmid P_2\mmid\ldots\mmid P_M,\,\wcP=P'_1\mmid P'_2\mmid\ldots\mmid P'_L$, and
$\wwcP=P''_1\mmid P''_2\mmid\ldots\mmid P''_R$.
Denote the Krawtchouk matrices of $(\cP,\wcP)$ and $(\wcP,\wwcP)$ by $K\in\C^{L\times M}$ and
$K'\in\C^{R\times L}$, respectively.
Then $(K'K)_{r,m}=|G|$ if $-P''_r\subseteq P_m$ and $(K'K)_{r,m}=0$ otherwise.
In particular, if $\cP$ is reflexive, then (after suitable ordering of the blocks) $K'K=|G|I_{M}$,
where $I_M$ is the identity matrix of size~$M$.
\end{prop}
\begin{proof}
Recall that we identify $\hat{\hat{G}}$ with~$G$.
Fix $g\in P''_r$.
Then
\[
  (K'K)_{r,m}=\sum_{\ell=1}^L\sum_{\chi\in P'_\ell}\inner{g,\chi}
                              \sum_{h\in P_m}\inner{\chi,h}     
              =\sum_{\ell=1}^L\sum_{\chi\in P'_\ell}\sum_{h\in P_m}\inner{g+h,\chi}
              =\sum_{h\in P_m}\sum_{\chi\in\hat{G}}\inner{g+h,\chi}.
\]
With the aid of~\eqref{e-GChar} we conclude that $(K'K)_{r,m}=|G|$ if $-g\in P_m$ and~$(K'K)_{r,m}=0$ otherwise.
\end{proof}

\section{Induced Partitions}\label{SS-IndPart}
In this section we turn to cartesian powers of groups.
We present two specific constructions of partitions and the resulting MacWilliams identities.
They cover many of the known MacWilliams in coding theory.

In the following we mostly write partitions in the form $(Q_\ell)_{\ell\in\cL}$, where~$\cL$ is a
suitable index set, rather than $Q_1\mmid\ldots\mmid Q_L$, due to a lack of a natural ordering of the blocks.

The two types of induced partitions defined next have been considered before in \cite[Sec.~7]{ZiEr09} and~\cite[Sec.~4.10]{Cam98}
for the product partition, and \cite[Sec.~2.5]{Del73} and \cite[Sec.~4.11]{Cam98} for the
symmetrized partition.

\begin{defi}\label{D-IndProdPart}
Let $\cG=G_1\times\cdots\times G_n$, where $G_1,\ldots,G_n$ are finite abelian groups and let
$\cP_i=P_{i,1}\mmid P_{i,2}\mmid\ldots\mmid P_{i,M_i}$ be partitions of~$G_i$ for $i=1,\ldots,n$.
On~$\cG$  we define the \emph{product partition} as
$\cP_1\times\ldots\times\cP_n:=(P_{1,m_1}\times\ldots\times P_{n,m_n})_{(m_1,\ldots,m_n)\in[M_1]\times\cdots\times[M_n]}$, where $[M_i]:=\{1,\ldots,M_i\}$.
If $G_i=G$ and $\cP_i=\cP$ for all~$i$, then $\cP\times\ldots\times\cP$ is written as~$\cP^n$ and called the
\emph{product partition of~$G^n$ induced by~$\cP$}.
\end{defi}

\begin{defi}\label{D-IndSymmPart}
Let $\cG=G^n$ and $\cP=P_1\mmid P_2\mmid\ldots\mmid P_M$ be a partition of the finite abelian group~$G$.
For $g=(g_1,\ldots,g_n)\in \cG$ define
\[
    \comp_{\cP}(g)=(s_1,\ldots,s_M), \text{ where }
   s_m=|\{t\mid g_t\in P_m\}|.
\]
We call $\comp_{\cP}(g)$ the \emph{composition vector} of~$g$ with respect to the partition~$\cP$.
It is contained in the set $\cS:=\{(s_1,\ldots,s_M)\in\N_0^M\mid \sum_{m=1}^M s_m=n\}$.
The \emph{induced symmetrized partition} of~$G^n$ is defined as
\begin{equation}\label{e-CompQ}
  \cPnsym=(Q_s)_{s\in\cS},\text{ where }Q_s=\{g\in\cG\mid \comp_{\cP}(g)=s\}.
\end{equation}
\end{defi}
Note that each block of the product partition~$\cP_1\times\ldots\times\cP_n$ consists of all $g\in G_1\times\cdots\times G_n$
for which each entry~$g_i$ is contained in a prescribed block of~$\cP_i$.
Obviously,~$\cP_1\times\ldots\times\cP_n$ has $M_1\cdot\ldots\cdot M_n$ blocks.
In particular,~$\cP^n$ has~$M^n$ blocks.
On the other hand, the blocks of the symmetrized partition~$\cPnsym$ collect all $g\in G^n$
that have the same number of entries (disregarding position) in a given block of~$\cP$.
The index set~$\cS$ is the set of weak $M$-partitions of~$n$, and
$\cP_{\text{sym}}^n$ consists of $|\cS|=\binom{n+M-1}{M-1}$, see~\cite[p.~15]{Sta97}.

As a simple example, if $\cP=\{0\}\mmidbig G\backslash\{0\}$, then~$\cP^n$ partitions the elements $(g_1,\ldots,g_n)\in G^n$
according to their support, whereas $\cPnsym$ classifies them with respect to their Hamming weight.

We show next that dualization commutes with the above constructions
under the rather weak condition that $\{0\}$ is a block of the given partitions.
Let us first consider an (extreme) example illustrating the necessity of this condition.
Suppose~$\cP$ consists of the single block~$G$.
Then both~$\cP^n$ and~$\cPnsym$ consist of the single block~$G^n$ and
$\widehat{\cP^n}=\widehat{\cPnsym}=\{\veps\}\mmidbig\hat{G}^n\backslash\{\veps\}$, see
Remark~\ref{R-DualNega}(a).
In particular, $\wcP=\{\veps\}\mmidbig\hat{G}\backslash\{\veps\}$.
Therefore, $\wcP^{\,n}$ consists of all $n$-fold product sets with factors $\{\veps\}$ and $\hat{G}\backslash\{\veps\}$,
and $\widehat{\cP}_{\text{sym}}^{\,n}$ is the Hamming partition on~$\hat{G}^n$.
All of this shows $\wcP^{\,n}\lneq\widehat{\cP^n}$ and
$\widehat{\cP}_{\text{sym}}^{\,n}\lneq\widehat{\cPnsym}$.

The second statement of the following theorem appears also in \cite[Thm.~4]{ZiEr09}  and \cite[Thm.~4.86]{Cam98}.
\begin{theo}\label{T-ProdPartDual}
Let~$\cP_i$ be a partition of~$G_i$ for $i=1,\ldots,n$ such that~$\{0\}$ is a block of~$\cP_i$ for all~$i$.
Let ~$\cQ=\cP_1\times\cdots\times\cP_n$.
Then $\widehat{\cQ}=\widehat{\cP}_1\times\cdots\times\wcP_n$.
As a consequence, if~$\cP_i$ is reflexive for all~$i$, then so is $\cP_1\times\cdots\times\cP_n$.
\end{theo}
\begin{proof}
Let $\cP_i$ be as in Definition~\ref{D-IndProdPart} and $\wcP_i=Q_{i,1}\mmid Q_{i,2}\mmid\ldots\mmid Q_{i,L_i}$.
With the aid of~\eqref{e-prodchar} we compute for all $\chi\in\hat{\cG}$
\begin{equation}\label{e-chiprod}
  \sum_{g\in P_{1,m_1}\times\ldots\times P_{n,m_n}}\hspace*{-1em}\inner{\chi,g}
  =\hspace*{-.2em}\sum_{g\in P_{1,m_1}\times\ldots\times P_{n,m_n}}\prod_{i=1}^n\inner{\chi_i,g_i}
  =\prod_{i=1}^n\sum_{g\in P_{i,m_i}}\hspace*{-.6em}\inner{\chi_i,g}.
\end{equation}
This shows immediately that if $\chi\widesim[2]_{\wcP_1\times\cdots\times\wcP_n}\chi'$ then $\chi\widesim[2]_{\widehat{\cQ}}\chi'$.
Considering~\eqref{e-chiprod} for all product partition blocks where $P_{i,m_i}=\{0\}$ for
all but one~$i$, establishes the converse.
\end{proof}

The analogous result is true for the induced symmetrized partition as well.
Again, for reflexive partitions the statement appears already in \cite[Sec.~2.5]{Del73} and \cite[Thm.~4.97]{Cam98} in
the terminology of abelian classes of association schemes.

\begin{theo}\label{T-SymmPartDual}
Let~$\cP$ be a partition of~$G$ such that~$\{0\}$ is a block of~$\cP$.
Then $\widehat{\cPnsym}=\widehat{\cP}_{\text{sym}}^{\,n}$.
As a consequence, if~$\cP$ is reflexive then so is $\widehat{\cPnsym}$.
\end{theo}

We postpone the quite technical proof  to Appendix~\ref{S-App}.
It reveals that establishing $\widehat{\cP}_{\text{sym}}^{\,n}\leq\widehat{\cPnsym}$ is basic and straightforward,
and the result for reflexive partitions follows immediately.
The general case, and thus the converse $\widehat{\cPnsym}\leq \widehat{\cP}_{\text{sym}}^{\,n}$, is significantly
more technical.
It makes use of the fact that the Krawtchouk coefficients appear as evaluations
of the elementary multi-symmetric polynomials.
We are not aware of a simpler, or any, proof in the literature.

\medskip

Now we derive the general MacWilliams identities for the induced partitions of~$G^n$
as defined in Definitions~\ref{D-IndProdPart} and~\ref{D-IndSymmPart}.
We restrict ourselves to the reflexive case, which simplifies notation as it allows us to define
the partition enumerators of~$\cC\leq\cG$ and~$\cC^\perp\leq\hat{\cG}$ in the same polynomial ring.
The general case can easily be dealt with, but as in Theorem~\ref{T-MacWPart} it requires the use of two
distinct polynomial rings and the careful choice of the partition and its dual.
For the reflexive case, analogous identities in the language of abelian
association schemes appear in~\cite[Thm.~5.46, Thm.~5.51]{Cam98}.

For the first result we assume the situation of Definition~\ref{D-IndProdPart} and fix the following notation.
Write $\cQ=\cP_1\times\ldots\times\cP_n$.
Let the dual partitions be $\wcP_i=Q_{i,1}\mid\ldots\mid Q_{i,M_i}$, and thus the blocks of
$\widehat{\cQ}$ are given by $Q_\ell=Q_{1,l_1}\times\ldots\times Q_{n,l_n}$
for all index combinations $\ell=(l_1,\ldots,l_n)$.
For $g\in G_i$ we denote by $m(g)\in\{1,\ldots,M_i\}$ the index of the block $P_{i,m(g)}$ containing~$g$,
and similarly for $\chi\in\hat{G_i}$ let $\chi\in Q_{i,m(\chi)}$.

\begin{theo}\label{T-weprodpart}
Let all partitions~$\cP_i$ be reflexive, thus~$\cQ$ is reflexive as well,
and let $K^{(i)}\in\C^{M_i\times M_i}$ be the Krawtchouk matrix of the pair $(\wcP_i,\cP_i)$.
For a code $\cC\leq\cG$ the polynomial
\[
  \pe_{\cQ,\,\cC}:=\sum_{g\in\cC}\prod_{i=1}^nX_{i,m(g_i)}\in W:=\C[X_{i,m}\mid i=1,\ldots,n,\,m=1,\ldots,M_i]
\]
is called the product partition enumerator of~$\cC$.
The coefficient of $\prod_{i=1}^nX_{i,m_i}$ equals the cardinality of $\cC\cap(P_{1,m_1}\times\ldots\times P_{n,m_n})$.
Similarly, the product partition enumerator of~$\cC^\perp\leq\hat{\cG}$ with respect to~$\widehat{\cQ}$ is
$\pe_{\widehat{\cQ},\,\cC^\perp}:=\sum_{\chi\in\cC^\perp}\prod_{i=1}^nX_{i,m(\chi_i)}\in W$.
The enumerators satisfy the MacWilliams identity
\begin{equation}\label{e-prodMacW}
  \pe_{\widehat{\cQ},\,\cC^{\perp}}=\frac{1}{|\cC|}\cM'(\pe_{\cQ,\,\cC}),
\end{equation}
where the MacWilliams transformation $\cM':\,W\longrightarrow W$ is the algebra
homomorphism satisfying
$\cM'(X_{i,m})=\sum_{\ell=1}^{M_i} K^{(i)}_{m,\ell}X_{i,\ell}$ for $i=1,\ldots,n$ and $m=1,\ldots,M_i$.
\end{theo}

\begin{proof}
Define the maps $\psi:\cG\rightarrow W,\ g\mapsto \prod_{i=1}^n X_{i,m(g_i)}$ and
$\bar{\psi}:\hat{\cG}\rightarrow W,\ \chi\mapsto \prod_{i=1}^n X_{i,m(\chi_i)}$.
Then $\pe_{\cQ,\cC}=\sum_{g\in\cC}\psi(g)$ and similarly
$\pe_{\widehat{\cQ},\cC^{\perp}}=\sum_{\chi\in\cC^{\perp}}\bar{\psi}(\chi)$.
By the definition of the Krawtchouk coefficients,  $K^{(i)}_{m,l}=\sum_{\chi\in Q_{i,l}}\inner{\chi,g}$ for any $g\in P_{i,m}$.
Using~\eqref{e-prodchar} we obtain for $g\in P_{1,m_1}\times\ldots\times P_{n,m_n}$
\begin{align*}
  \bar{\psi}^+(g)&=\sum_{\chi\in\hat{\cG}}\inner{\chi,g}\bar{\psi}(\chi)
    =\sum_{\ell=(\ell_1,\ldots,\ell_n)}\sum_{\chi\in Q_\ell}\inner{\chi,g}\prod_{i=1}^n X_{i,\ell_i}
    =\sum_{\ell=(\ell_1,\ldots,\ell_n)}\sum_{\chi\in Q_\ell}\prod_{i=1}^n \inner{\chi_i,g_i}X_{i,\ell_i}\\
    &=\prod_{i=1}^n \sum_{\ell=1}^{M_i}\sum_{\chi_i\in Q_{i,\ell}}\inner{\chi_i,g_i}X_{i,\ell}
     =\prod_{i=1}^n\sum_{\ell=1}^{M_i} K^{(i)}_{m_i,\ell}X_{i,\ell}=\cM'\big(\prod_{i=1}^n X_{i,m_i}\big).
\end{align*}
Wit the aid of the Poisson summation formula~\eqref{e-Poisson} we derive
\[
  \cM'(\pe_{\cQ,\cC})=\cM'\big(\sum_{g\in\cC}\prod_{i=1}^n X_{i,m(g_i)}\big)
   =\sum_{g\in\cC}\bar{\psi}^+(g)=|\cC|\sum_{\chi\in\cC^{\perp}}\bar{\psi}(\chi)
   =|\cC|\pe_{\widehat{\cQ},\cC^{\perp}}.
   \qedhere
\]
\end{proof}

Notice that we may write the identity~\eqref{e-prodMacW} in the form
\begin{equation}\label{e-MacWProd}
   \pe_{\widehat{\cQ},\,\cC^{\perp}}({\bf X}_i \mid i=1,\ldots,n)=
   \frac{1}{|\cC|}\pe_{\cQ,\,\cC}(K^{(i)}{\bf X}_i \mid i=1,\ldots,n),
\end{equation}
where ${\bf X}_i=(X_{i,1},\ldots,X_{i,M_i})\T$.

\medskip
In the same way we can derive a MacWilliams identity for the induced symmetrized partitions of~$G^n$.
Again, we restrict ourselves to reflexive partitions.
As before, let $m(g)$ be the unique index such that $g\in P_{m(g)}$.

\begin{theo}\label{T-wesymmpart}
Let $\cP=P_1\mid\ldots\mid P_M$ be a reflexive partition of~$G$, and let $K\in\C^{M\times M}$ be the
Krawtchouk matrix of $(\wcP,\cP)$.
For a code $\cC\leq G^n$ the symmetrized partition enumerator of~$\cC$ with respect to~$\cP$ is defined as
$\pe_{\cP^n_{\text{sym}},\,\cC}:=\sum_{g\in\cC}\prod_{t=1}^nY_{m(g_t)}$ and is contained in
the polynomial ring $V:=\C[Y_m\mid m=1,\ldots,M]$.
It is a homogeneous polynomial of degree~$n$, and the coefficients of the monomial
$\prod_{m=1}^M Y_{m}^{s_m}$  equals the cardinality
$|\{g\in\cC\mid \comp_{\cP}(g)=(s_1,\ldots,s_M)\}|$.
The enumerator satisfies the MacWilliams identity
\begin{equation}\label{e-symmMacW}
  \pe_{\wcP^n_{\text{sym}},\,\cC^{\perp}}=\frac{1}{|\cC|}\cM''(\pe_{\cP^n_{\text{sym}},\,\cC}),
\end{equation}
where the MacWilliams transformation $\cM'':\,V\longrightarrow V$ is the algebra
homomorphism given by $\cM''(Y_m)=\sum_{\ell=1}^M K_{m,\ell}Y_\ell$ for $m=1,\ldots,M$.
\end{theo}

\begin{proof}
We make use of the MacWilliams identity for the product partition enumerator in Theorem~\ref{T-weprodpart}, applied to the
product partition~$\cP^n$ and its dual.
Then the polynomial ring~$W$ is $W=\C[X_{i,m}\mid i=1,\ldots,n,\,m=1,\ldots,M]$.
Consider the substitution homomorphism
$\rho: W\longrightarrow V$, $X_{t,m}\longmapsto Y_m$.
Then Theorem~\ref{T-weprodpart} yields
\[
 \pe_{\wcP^n_{\text{sym}},\,\cC^{\perp}}=\rho(\pe_{\wcP^n,\cC^{\perp}})=\frac{1}{|\cC|}\rho\circ\cM'(\pe_{\cP^n,\,\cC}).
\]
Note that $\rho\circ\cM'(X_{t,m(g)})=\rho\big(\sum_{\chi\in\hat{G}}\inner{\chi,g}X_{t,m(\chi)}\big)
  =\sum_{\chi\in\hat{G}}\inner{\chi,g}Y_{m(\chi)}=\cM''(Y_{m(g)})
  =\cM''\circ\rho(X_{t,m(g)})$.
Hence  $\rho\circ\cM'(\pe_{\cP^n,\,\cC})=\cM''\circ\rho(\pe_{\cP^n,\,\cC})=\cM''(\pe_{\cP^n_{\text{sym}},\,\cC})$, and this concludes the proof.
\end{proof}

As we did for the product partition, we may write the identity~\eqref{e-symmMacW} in the form
\begin{equation}\label{e-MacWSymm}
   \pe_{\wcP^n_{\text{sym}},\,\cC^{\perp}}({\bf Y})=
   \frac{1}{|\cC|}\pe_{\cP^n_{\text{sym}},\,\cC}\big(K{\bf Y}\big),
\end{equation}
where ${\bf Y}=(Y_1,\ldots,Y_M)\T$.

\medskip
The last two theorems cover an abundance of MacWilliams identities from the literature: the identities for the Hamming weight,
the complete weight~\cite{MS77,Kl89,Cam98}, the exact weight~\cite{MS77,Cam98}, the symmetrized Lee weight~\cite{HKCSS94,Kl87}
are all instances of Theorem~\ref{T-wesymmpart} for the symmetrized partition enumerator, and so are many other cases.
The only work left is the explicit computation of the Krawtchouk coefficients in each concrete case.
But for the just mentioned examples, this can be done straightforwardly; see also \cite[Sec.~4]{GL13MacW1}, where this
has been carried out in detail.

While these cases are well known and can be found in the above mentioned literature, we wish to touch upon a different
class of lesser known identities explicitly.
That is those of split weight enumerators, where for instance, the Hamming weight is considered separately
on various components of a given vector, or the Hamming weight is considered on one part and the symmetrized
Lee weight on the other one.
All these cases are instances of Theorem~\ref{T-weprodpart}.
The resulting MacWilliams identity for the \emph{split Hamming weight enumerator} has been derived
by MacWilliams and Sloane for codes over the binary field  in~\cite[Ch.~5, Eq.~(52)]{MS77} and by
Simonis~\cite[Eq.~(3')]{Si95} for arbitrary fields and where the codewords are divided into~$t$ blocks of coordinates.
A similar identity can be found in~\cite{KhMcE05} by El-Khamy and McEliece.
The latter authors also observe that if~$\cC$ is a systematic~$[n,k]$ code, then the split Hamming weight enumerator
is the input-redundancy weight enumerator which keeps track of the input weights in combination with the corresponding
redundancy weight.
This allows them to apply their identity to MDS codes in order to derive further results on the bit error probability
for systematic RS codes.
Finally, in~\cite{LKY04} this weight enumerator has been used to derive a MacWilliams identity for the input-output
weight enumerators of direct-product single-parity-check codes.

We will apply Theorems~\ref{T-weprodpart} and~\ref{T-wesymmpart} at the end of Section~\ref{SS-Posets} to the
particular partitions discussed in that section.

\section{Duality for Frobenius Rings}\label{SS-FrobRing}
In this section we focus our attention on the case where the group~$G$ is the additive group of a finite ring.
We restrict ourselves to commutative rings in order to keep notation simple and because most known and interesting
examples are for codes over commutative rings.
The results of this section are not new, and the goal is rather to carefully reconcile duality in the group setting with
that for codes over rings.
For a Frobenius ring~$R$, the character group~$\hat{R}$ can be turned into a module which is isomorphic
to the given ring.
We will illustrate that when identifying these two modules, as often done in the literature of codes over rings,
the dual of a partition depends on the identification.

Let~$R$ be a finite commutative ring with identity. Its group of units is denoted by~$R^*$.
The character group of $(R,\,+\,)$ can be endowed with an $R$-module structure via the scalar multiplication
$r\chi(a):=\chi(ra)$, and we call~$\hat{R}$ the \emph{character module} of~$R$.

While the additive groups of~$R$ and~$\hat{R}$ are isomorphic, this is not necessarily the case for
the $R$-modules~$R$ and~$\hat{R}$.
The latter are isomorphic if and only if the ring is Frobenius.
In ring theory, Frobenius rings are commonly defined via their socle, see~\cite[Def.~16.14]{Lam99}.
For finite commutative rings, however, it follows from Lamprecht~\cite{Lamp53} (see also Hirano~\cite[Thm.~1]{Hi97} and 
Honold~\cite[p.~409]{Hon01}) that this is equivalent to our definition below.
Since this character-theoretic property is all we need in this paper, we simply use this as our definition.

\begin{defi}\label{D-Frob}
A finite commutative ring~$R$ is called \emph{Frobenius} if there exists a character~$\chi\in\hat{R}$ such that
$\alpha:\,R\longrightarrow\hat{R},\ r\longmapsto r\chi$ is an $R$-isomorphism.
Any character~$\chi$ with this property is called a \emph{generating character} of~$R$.
\end{defi}
Obviously, any two generating characters~$\chi,\,\chi'$ differ by a unit, i.e., $\chi'=u\chi$ for some
$u\in R^*$.



Many standard examples of commutative rings are Frobenius.
Details can be found in Wood~\cite[Ex.~4.4]{Wo99} and Lam~\cite[Sec.~16.B]{Lam99}.
\begin{exa}\label{E-Frob}
\begin{alphalist}
\item The integer residue rings $\Z_N$, where $N\in\N$, are Frobenius; see Example~\ref{E-Z6Z8}(a).
\item Every finite field is Frobenius, and every non-principal character of~$\F$ is a generating character.
\item Finite chain rings, finite group rings over a Frobenius ring, direct products of Frobenius rings, and
        Galois rings are Frobenius.
\item The ring $R=\F_2[x,y]/(x^2,y^2,xy)$ is a local, non-Frobenius ring; see~\cite[Ex.~3.2]{ClGo92}.
      We will also recover this result below in Remark~\ref{R-DoubleAnn}.
\end{alphalist}
\end{exa}

The following easy-to-verify property has been proved by Claasen and Goldbach~\cite[Cor.~3.6]{ClGo92}.
\begin{rem}\label{R-PropGen}
Let~$\chi$ be a character of~$R$.
Then~$\chi$ is a generating character of~$R$ if and only if the only ideal contained in
$\ker\chi:=\{r\in R\mid \chi(r)=1\}$ is the zero ideal.
\end{rem}

Now we can derive the following familiar identifications. For part~(b) see also Wood~\cite[Thm.~7.7]{Wo99}.
A \emph{(linear) code} over~$R$ is simply a submodule of $R^n$.

\begin{theo}\label{T-RFrob}
Let~$R$ be a finite commutative Frobenius ring and~$\chi$ a generating character of~$R$.
For $v,a\in R^n$ denote by $v\cdot a:=\sum_{i=1}^nv_ia_i$ the dot product.
Moreover, for $v\in R^n$ define the character $\chi_v\in\widehat{R^n}$ via $\chi_v(a)=\chi(v\cdot a)=\inner{\chi,v\cdot a}$ for $a\in R^n$.
Then we have the following.
\begin{alphalist}
\item The map $\alpha:\;R^n\longrightarrow\widehat{R^n},\quad v\longmapsto\chi_v$ is an $R$-module
      isomorphism.
\item For a code $\cC\subseteq R^n$ define the dot-product dual as
      $\cC^{\pperp}:=\{v\in R^n\mid v\cdot a=0\text{ for all }a\in\cC\}$.
      Then the character-theoretic dual of the additive group $(\cC,\,+\,)$ and the dot-product dual coincide; precisely
      $\alpha(\cC^{\pperp})=\cC^\perp=\{\psi\in\widehat{R^n}\mid \inner{\psi,a}=1\text{ for all }a\in\cC\}$.
\end{alphalist}
\end{theo}
\begin{proof}
Recall from~\eqref{e-prodchar} that~$\hat{R}^n$ and~$\widehat{R^n}$  are isomorphic groups via
$\inner{(\chi_1,\ldots,\chi_n),(a_1,\ldots,a_n)}=\prod_{i=1}^n\inner{\chi_i,a_i}$ for all $\chi_i\in\hat{R}$.
For the Frobenius ring~$R$, all characters are of the form $r\chi,\, r\in R$, and thus the last identity reads as
\[
   \inner{(r_1\chi,\ldots,r_n\chi),(a_1,\ldots,a_n)}=\prod_{i=1}^n\inner{\chi, r_ia_i}
   =\inner{\chi,\sum_{i=1}^nr_ia_i}=\inner{\chi, r\cdot a}.
\]
This proves~(a).
As for~(b), notice first that the dual group~$\cC^\perp$ is indeed an $R$-module.
Next, the containment $\alpha(\cC^{\pperp})\subseteq\cC^\perp$ is evident.
For the converse, let $v\in R^n$ be such that~$\alpha(v)\in\cC^\perp$.
Thus $\inner{\chi_v,a}=1$ for all $a\in\cC$.
But then for all $r\in R$ and  all $a\in\cC$ we have $1=\inner{\chi_v, ra}=\inner{\chi, rv\cdot a}$.
This means that the ideal in~$R$ generated by $v\cdot a$ is in $\ker\chi$, and with the aid of
Remark~\ref{R-PropGen} we conclude $v\cdot a=0$.
Since $a\in\cC$ is arbitrary this shows that $v\in\cC^{\pperp}$.
\end{proof}

We have the following simple, but crucial consequence of~\eqref{e-bidual}.
\begin{rem}\label{R-DoubleAnn}
For any Frobenius ring~$R$ and any code $\cC\subseteq R^n$, we have $\cC^{\pperp\pperp}=\cC$.
As a consequence, $|\cC||\cC^{\pperp}|=|R^n|$.
If $n=1$, thus $\cC\subseteq R$ is an ideal in~$R$, then the identity $\cC^{\pperp\pperp}=\cC$ is known as the
double annihilator property, see, e.g.,~\cite[Thm.~15.1]{Lam99}.
This property is in general not true if~$R$ is not a Frobenius ring as can easily be seen using
the ideal $\cC=(x)$ in the ring~$R=\F_2[x,y]/(x^2,y^2,xy)$.
In this case, $(x)^{\pperp}=(x,y)$ and $(x)^{\pperp\pperp} =(x,\,y)\neq(x)$  (thereby proving that~$R$ is not Frobenius).
%
\end{rem}

It is important to keep in mind that the identification of~$R^n$ and~$\widehat{R^n}$ depends on the choice
of the generating character.
As a consequence, the dual of a partition of~$R^n$, if considered in~$R^n$ again, may also depend on the
generating character.
An example will be given after Definition~\ref{D-DualPartchi}.
For residue rings~$R=\Z_N$ this dependence does not occur.
This is due to the fact that all primitive  $N$-th roots of unity in~$\C$ have the same minimal polynomial.
To be on the safe side we cast the following definition.

\begin{defi}\label{D-DualPartchi}
Let~$\chi$ be a generating character of the Frobenius ring~$R$ and $\alpha:R^n\rightarrow\widehat{R^n}$ be the
isomorphism $v\mapsto\chi_v$ from Theorem~\ref{T-RFrob}(a).
For a partition $\cP=P_1\mmid P_2\mmid\ldots\mmid P_M$  of~$R^n$ and its dual partition~$\wcP$ of~$\widehat{R^n}$
we define the $\chi$-dual partition~$\wcP^{^{[\chi]}}$ of~$R^n$ as $\alpha^{-1}(\wcP)$.
Thus
\[
   v\widesim[2]_{\wcP^{^{[\chi]}}}v'\Longleftrightarrow \sum_{w\in P_m}\chi(v\cdot w)=\sum_{w\in P_m}\chi(v'\cdot w)
   \text{ for all }m=1,\ldots,M.
\]
The partition~$\cP$ is called \emph{$\chi$-self-dual} if $\wcP^{^{[\chi]}}=\cP$.
\end{defi}

Here comes an example illustrating the dependence of the dualization on the choice of the generating character.
\begin{exa}\label{E-FPartNotIndep}
Consider the field~$\F_4=\{0,\,1,\,a,\,a^2\}$.
The maps~$\chi$ such that $\chi(0)=\chi(1)=1,\,\chi(a)=\chi(a^2)=-1$ and $\tilde{\chi}$ such that
$\tilde{\chi}(0)=\tilde{\chi}(a)=1,\,\tilde{\chi}(1)=\tilde{\chi}(a^2)=-1$
are characters of~$\F_4$.
The partition~$\cP=0\mmid1\mmid a,a^2$ of~$\F_4$ satisfies
$\wcP^{^{[\chi]}}=\cP$, whereas $\wcP^{^{[\tilde{\chi}]}}=0\mmid1,a^2\mmid a$.
Hence $\cP$ is $\chi$-self-dual but not $\tilde{\chi}$-self-dual.
\end{exa}

\begin{rem}\label{R-bidual}
It is easy to verify that the bidual partition does not depend on the generating character~$\chi$.
Precisely, $\widehat{\wcP^{^{[\chi]}}}^{^{[\chi]}}=\wwcP$ for any generating character~$\chi$.
To see this, let $\cQ=Q_1\mmid\ldots\mmid Q_L =\wcP^{^{[\chi]}}$ and let $\cR=\widehat{\cQ}^{^{[\chi]}}$.
Then by definition $v\widesim[2]_{\cR}v'\Leftrightarrow \alpha(v)\widesim_{\widehat{\cQ}}\alpha(v')$ for any $v,\,v'\in R^n$.
The computation
\[
  \sum_{w\in Q_\ell}\inner{\alpha(v),w}=\sum_{w\in Q_\ell}\inner{\chi_v,w}=\sum_{w\in Q_\ell}\chi(v\cdot w)
  =\sum_{w\in Q_\ell}\inner{\alpha(w),v}=\sum_{\tilde{\chi}\in\alpha(Q_\ell)}\inner{\tilde{\chi},v}
\]
along with $\alpha(\cQ)=\wcP$ shows that $\cR=\wwcP$.
\end{rem}

\medskip
We close this section with a brief digression and comment on a different approach to duality and MacWilliams identities taken by
Honold and Landjev.
In~\cite{HoLa01} they study codes~$\cC$ in~$R^n$, where~$R$ is a (not necessarily commutative)
Frobenius ring.
The dual is defined as the dot product dual~$\cC^{\pperp}$.
Instead of using characters, the authors make use of a (unique) map $\omega:R\rightarrow\Q$ with certain ``homogeneous'' properties.
The existence of such a map is guaranteed for Frobenius rings.
The authors define two particular classes of pairs of partitions of~$R^n$: regular pairs and ${\mathcal W}$-admissible pairs.
It is easy to see that if a pair of partitions $(\cP,\cQ)$ of~$R^n$ is regular, then $\wcP=\cQ$ and $\widehat{\cQ}=\cP$
(where we identify~$R^n$ with~$\widehat{R^n}$ as in Theorem~\ref{T-RFrob}(a)).
But ${\mathcal W}$-admissible pairs are unrelated to our notion of reflexivity: one can construct
${\mathcal W}$-admissible partitions that are not reflexive and vice versa.
In~\cite[Thm.~21]{HoLa01} the authors show that ${\mathcal W}$-admissible pairs allow a MacWilliams identity
analogous to the one in Theorem~\ref{T-MacWPart} of this paper.
Both reflexivity and ${\mathcal W}$-admissibility have their advantages as either property carries over to certain derived partitions (e.g., the  product partition and symmetrized product partition for reflexitivity),
and just like the approach with reflexive partitions  the one taken in~\cite{HoLa01} leads to plenty of
interesting applications.

\section{Poset Structures on~$G_1\times\ldots\times G_n$}\label{SS-Posets}
Let~$G_i$ be non-trivial finite abelian groups for $i=1,\ldots,n$.
In this section we consider a weight for codes in~$\cG:=G_1\times\ldots\times G_n$ that derives from
a prescribed poset structure on the coordinate set $[n]:=\{1,\ldots,n\}$.
This has been introduced by Brualdi et al. in~\cite{BGL95} for codes over fields (i.e., codes in $\F^n$), where it generalizes the Hamming weight as well as the Rosenbloom-Tsfasman weight.
The latter has been introduced by Rosenbloom and Tsfasman in~\cite{RoTs97} and plays a specific role for matrix codes; see~\cite{Skr07} for the relevance of the Rosenbloom-Tsfasman
weight for detecting matrix codes with large Hamming distance.

In~\cite{DoSk02} Dougherty and Skriganov establish a MacWilliams identity for the Rosenbloom-Tsfasman weight
and a suitably defined dual weight for codes over fields.
This has been generalized by Kim and Oh~\cite{KiOh05} to general poset structures, and they also characterize the poset structures that allow such a MacWilliams identity.
Pinheiro and Firer~\cite{PiFi12} generalize this even further to block poset structures.
In all these cases the dual weight is induced by the dual poset.

In this section, we will recover these MacWilliams identities and generalize them to codes over groups.
We will do so by establishing that the partitions induced by a poset weight and the dual poset weight
are mutually dual in the sense of Definition~\ref{D-FPart}.
Then the MacWilliams identity is simply an instance of Theorem~\ref{T-MacWPart}.

Let $\leq$ be a partial order on $[n]$, thus $\bP:=([n],\leq)$ is a poset.
We denote by $\max(\bP)$ (resp.\ $\min(\bP)$) the set of all maximal (resp.\ minimal) elements of~$\bP$.
A subset $S\subseteq[n]$ is called an ideal if $i\in S$ and $j\leq i$ implies $j\in S$.
Denote by $\ideal{S}$ the smallest ideal generated by the set~$S$.
Let~$j\in\max(\bP)$ and consider the subposet $\bP'=([n]\backslash\{j\},\leq)$.
Then evidently,
\begin{equation}\label{e-idealPP'}
  \ideal{S}_{\bP}=\ideal{S}_{\bP'}\ \text{ for any subset } S\subseteq[n]\backslash\{j\}.
\end{equation}

A poset~$\bP=([n],\leq)$ induces the \emph{poset weight} on~$\cG$  given by
\begin{equation}\label{e-pweight}
   \wt_{\bP}(g)=\big|\ideal{\supp(g)}\big|,
\end{equation}
where, as usual, $\supp(g):=\{i\mid g_i\neq0\}$ denotes the support of~$g=(g_1,\ldots,g_n)\in\cG$.
The weight induces a metric on~$\cG$ (see~\cite[Lem.~1.1]{BGL95} for~$\F^n$).
More interesting to us, it gives rise to a partition~$\cP_{\bP}=P_0\mmid\ldots\mmid P_n$ of~$\cG$ via
\[
   P_m=\{g\in \cG\mid \wt_{\bP}(g)=m\} \text{ for } m=0,\ldots,n.
\]
The blocks~$P_m$ are nonempty for each~$m$ and thus $|\cP_{\bP}|=n+1$.
Indeed,~$P_n\not=\emptyset$ is obvious and for~$P_{n-1}$ pick an element in~$\cG$ whose
support is given by $[n]\backslash\{j\}$ for some $j\in\max(\bP)$.
Now the rest follows inductively using~\eqref{e-idealPP'}.

A particular role is played by hierarchical posets.
For the terminology we follow Kim and Oh~\cite{KiOh05}.
\begin{defi}\label{D-hierposet}
Let $\bP=([n],\leq)$ be a poset. Then $\bP$ is called a \emph{hierarchical poset} if there exists a partition
$[n]=\bigcup_{i=1}^t\hspace*{-2.1em}\raisebox{.5ex}{$\cdot$}\hspace*{1.8em} \Gamma_i$
such that for all $l,m\in[n]$ we have $l< m$ if and only if  $l\in\Gamma_{i_1},\,m\in\Gamma_{i_2}$
for some $i_1<i_2$ (where $i_1<i_1$ refers to the natural order in~$\N$).
In other words, for all $i_1<i_2$ every element in $\Gamma_{i_1}$ is less than every element in $\Gamma_{i_2}$,
and no other two distinct elements in~$[n]$ are comparable.
We call $\Gamma_{i}$ the \emph{$i$-th level} of~$\bP$.
The hierarchical poset is completely determined (up to order-isomorphism)
by the data $(n_1,\ldots,n_t)$, where $n_i=|\Gamma_i|$, and it is
denoted by $\bH(n;n_1,\ldots,n_t)$.
\end{defi}


\begin{exa}\label{E-Poset1}
\begin{alphalist}
\item An \emph{anti-chain} on~$[n]$ is a poset in which any two distinct elements in~$[n]$ are incomparable.
        Thus,~$\bP$ is an anti-chain  if and only if~$\bP$ is the hierarchical poset $\bH(n;n)$.
        In this case, $\ideal{\supp(v)}=\supp(v)$ for all $g\in\cG$, and $\wt_{\bP}$ is
        simply the Hamming weight on~$\cG$.
\item A poset~$\bP=([n],\leq)$  is a \emph{chain} if~$\leq$ is a total order.
        Thus,~$\bP$ is a chain if and only if $\bP$ is the hierarchical poset $\bH(n;1,\ldots,1)$.
        Assuming without loss of generality that $1<2<\ldots<n$, we observe that
         $\wt_{\bP}(g)=\max\{i\mid g_i\neq0\}$, which on~$\F^n$ is known as the
         \emph{Ro\-sen\-bloom-Tsfasman weight}.
\end{alphalist}
\end{exa}

The dual of the poset~$\bP$ is defined as the poset $\bbP=([n],\geq)$ where
$x\geq y:\Longleftrightarrow y\leq x$.
The following is immediate.
\begin{rem}\label{R-dualhier}
Let~$\bP$ be a hierarchical poset with~$t$ levels.
Then the dual~$\bbP$ is a hierarchical poset with $t$ levels.
Its $\ell$-th level is the $(t+1-\ell)$-th level of~$\bP$.
\end{rem}

For an inductive argument later on it is convenient to consider a particular situation beforehand.

\begin{rem}\label{R-levels}
Let $\bP=([n], \leq)$ a poset with dual poset $\bbP$.
Let $\max(\bP)=\{l+1,\ldots,n\}$, which then is also $\min(\bbP)$.
Suppose that for all $i\in\max(\bP)$ and all $j\in [n]\backslash\max(\bP)$ we have $j<i$.
Denote by ${\mathbf Q}=([l],\leq)$  the restriction of the partial order~$\leq$ on~$[l]$.
Then for any $g=(g_1,\ldots,g_n)\in\cG$ and $\chi=(\chi_1,\ldots,\chi_n)\in\hat{\cG}$
\begin{align*}
  \wt_{\bP}(g)&=\left\{\begin{array}{cl}
    \wt_{\mathbf Q}(g_1,\ldots,g_l),& \text{if }(g_{l+1},\ldots,g_n)=0,\\
    l+\wt_{\text{H}}(g_{l+1},\ldots,g_n),&\text{otherwise,}\end{array}\right.\\[1ex]
  \wt_{\bbP}(\chi)&=\left\{\begin{array}{cl}
        n-l+\wt_{\bar{\mathbf Q}}(\chi_1,\ldots,\chi_l),&\text{if }(\chi_{1},\ldots,\chi_l)\neq\veps,\\
        \wt_{\text{H}}(\chi_{l+1},\ldots,\chi_n), &\text{otherwise,}\end{array}\right.
\end{align*}
where~$\wt_{\text{H}}$ is the Hamming weight.
(Recall from Example~\ref{E-Z6Z8}(c) that for characters in~$\hat{\cG}$  the Hamming weight~$\wt_{\text{H}}$
is the number of entries not equal to the principal character~$\veps$.)
Hence restricting the elements of the blocks~$P_m,\,\bar{P}_m$ of the partitions
$\cP_{\bP},\,\cP_{\bbP}$ to the index set~$[l]$ leads to the following relation with
the blocks $Q_m,\,\bar{Q}_m$ of $\cP_{\mathbf Q},\,\cP_{\bar{\mathbf Q}}$:
\[
   (P_m)\raisebox{-1ex}{\big|}_{[l]}=Q_m\text{ and }(\bar{P}_{m+n-l})\raisebox{-1ex}{\big|}_{[l]}=\bar{Q}_m\text{ for all }m=0,\ldots,l.
\]

\end{rem}

In the following remark we fix convenient index notation for hierarchical posets and for the factors of the group~$\cG$.
It leads to a simple formula for the poset weight.

\begin{rem}\label{R-HierPoset}
Let $\bP$ be the hierarchical poset $\bH(n;n_1,\ldots,n_t)$.
Note that $\sum_{i=1}^t n_i=n$.
We write the underlying set~$[n]$ as
\[
       \cN:=\{(i,j)\mid i=1,\ldots,t,\,j=1,\ldots,n_i\}
\]
such that the partial order simply reads as
$(i,j)< (i',j')\Longleftrightarrow i<j$
(where~$i<j$ refers to the natural order on~$\N$).
Consequently, for any set $S=\{(i_1,j_1),\ldots,(i_r,j_r)\}\subseteq\cN$, where $i_1\leq\ldots\leq i_{s-1}<i_s=i_{s+1}=\ldots=i_r$,
the ideal generated by~$S$ is
\[
  \ideal{S}=\{(i,j)\in\cN\mid i< i_r\}\cup\{(i_s,j_s),(i_{s+1},j_{s+1}),\ldots,(i_r,j_r)\}.
\]
Accordingly, we index the factors of the group~$\cG$ as $G_{i,j},\,i=1,\ldots,t,\,j=1,\ldots,n_i$.
Thus,
\begin{equation}\label{e-Gfactors}
     \cG=\prod_{i=1}^t\cG_i, \text{ where }\cG_i=\prod_{j=1}^{n_i} G_{i,j},
\end{equation}
and $g\in\cG$ has the form $g=(g_1,\ldots,g_t)$, where $g_i=(g_{i,1},\ldots,g_{i,n_i})\in\cG_i$.
Now the definition of the poset weight yields
\begin{equation}\label{e-HPweight}
         \wt_{\bP}(g)=\sum_{i=1}^{s-1}n_i+\wt_{\text{H}}(g_s),\text{ where }s=\max\{i\mid g_i\not=0\}.
\end{equation}
Here $\wt_{\text{H}}$ stands for the Hamming weight on each of the groups~$\cG_i$.
Similarly, the dual poset weight on the character group $\hat{\cG}=\hat{\cG_1}\times\ldots\times\hat{\cG_t}$ is given by
\begin{equation}\label{e-HPPweight}
         \wt_{\bbP}(\chi)=\sum_{i=s+1}^{t}n_i+\wt_{\text{H}}(\chi_s),\text{ where }s=\min\{i\mid \chi_i\not=\veps\}.
\end{equation}
\end{rem}

For codes in~$\F^n$ (thus $G_i=(\F,+)$ for all~$i$), Kim and Oh~\cite{KiOh05}  studied the question which posets give rise to a MacWilliams identity for the induced $\wt_{\bP}$-enumerator of a code and the $\wt_{\bbP}$-enumerator of its dual.
They proved that this is the case if and only if the poset is hierarchical.
Pinheiro and Firer~\cite{PiFi12} generalized this result to block poset structures.
We will recover both these results and generalize them to codes over groups.
This will be accomplished by relating the induced poset weight partitions to reflexivity in the sense of
Definition~\ref{D-FPart}.
Here is the first part of our result.

\begin{theo}\label{T-DualPartPoset}
Fix a poset $\bP=([n],\leq)$.
Let $\cP_{\bP}$ and~$\cP_{\bbP}$ be the induced partitions on $\cG=G_1\times\ldots\times G_n$ and
$\hat{\cG}=\hat{G_1}\times\ldots\times\hat{G_n}$, respectively.
Suppose $\widehat{\cP_{\bP}}=\cP_{\bbP}$.
Then~$\bP$ is a hierarchical poset, and if~$\bP=\bH(n;n_1,\ldots,n_t)$ and $\cG$ is as in~\eqref{e-Gfactors}, then
 $|G_{i,1}|=\ldots=|G_{i,n_i}|$ for all~$i=1,\ldots,t$.
 \end{theo}

Note that the assumption $\widehat{\cP_{\bP}}=\cP_{\bbP}$ along with $|\cP_{\bP}|=|\cP_{\bbP}|=n+1$ implies that
$\cP_{\bP}$ is reflexive; see Theorem~\ref{P-LPplus}.
But one should be aware that Theorem~\ref{T-DualPartPoset} does not state that reflexivity of the poset partition~$\cP_{\bP}$
implies that~$\bP$ is a hierarchical poset.
We strongly believe that this is true as well, but unfortunately do not have a proof.
To be more specific, supported by many examples we believe that
$\widehat{\cP_{\bP}}\leq\cP_{\bbP}$ for any poset~$\bP$.
This would imply the above conjecture.

One may notice that by Theorem~\ref{T-MacWPart} the assumption $\widehat{\cP_{\bP}}=\cP_{\bbP}$ yields the
existence of a MacWilliams identity between the corresponding weight enumerators.
As a consequence, if $\cG=\F^n$ or  $\cG=\F^{k_1}\times\ldots\times\F^{k_n}$, the above statement follows from
Kim and Oh~\cite[Thm.~2.5]{KiOh05} or Pinheiro and Firer~\cite[Thm.~1]{PiFi12};
see also the discussion after Corollary~\ref{C-HposetKraw}.
Their proofs for linear codes over fields, however, do not carry over to additive codes.

\begin{proof}[Proof of Theorem~\ref{T-DualPartPoset}]
Let $\cP_{\bP}=P_0\mmid P_1\mmid\ldots\mmid P_n$ and
$\cP_{\bbP}=\bar{P}_0\mid\bar{P}_1\mmid \ldots\mmid\bar{P}_n$, where $P_0=\{0\}$ and $\bar{P}_0=\{\veps\}$.
Suppose first that~$\bP$ is not hierarchical.
Then~$\bbP$ is not hierarchical either.
We use $<_{\bar{\bf P}}$ and $\ideal{\ }_{\bar{\bf P}}$ for the dual partial order and the ideals in the dual poset.
Without loss of generality let $\min(\bbP)=\{1,\ldots,l\}$ for some $l\in\{1,\ldots,n\}$.
We may assume that there exists~$j$ in
$[n]\backslash\min(\bbP)$ and $i\in\min(\bbP)$ such that $i\not<_{\bar{\bf P}} j$.
(Otherwise we may disregard $\min(\bbP)$ and proceed with the non-hierarchical poset
on $\{l+1,\ldots,n\}$ since by~Remark~\ref{R-levels} the resulting partition sets are restrictions of
blocks of~$\cP_{\bP}$ and $\cP_{\bbP}$.)
We may choose~$j$ to be a minimal element in~$\bbP$ with the above specified property.
Then $\ideal{j}_{\bar{\bf P}}\subseteq\min(\bbP)\cup\{j\}\backslash\{i\}$, and thus $|\ideal{j}_{\bar{\bf P}}|=:m\leq l$.
Choose $\chi'\in\hat{\cG}$ such that $\supp(\chi')\subseteq\min(\bbP)$ and $|\supp(\chi')|=m$.
Moreover, let $\chi''=(\veps,\ldots,\veps,\chi''_j,\veps,\ldots,\veps)\in\hat{\cG}$ be such that
its $j$-th entry~$\chi''_j$ is not the principal character of~$G_j$.
Then $\chi',\,\chi''\in\bar{P}_m$.
We show that $\chi'\not\hspace*{-.7em}\widesim[2]_{\widehat{\cP_{\bP}}}\chi''$ and thus
$\widehat{\cP_{\bP}}\not=\cP_{\bbP}$.
The block~$P_n$ of~$\cP_{\bP}$ is evidently of the form (see also \cite[Lem.~2.1]{KiOh05})
\[
   P_n\!=\!\{g\in\cG\!\mid\! \max(\bP)\subseteq\supp(g)\} =\{g\in\cG\!\mid\! \min(\bbP)\subseteq\supp(g)\}
         \!=\!\{(g_1,\ldots,g_n)\!\mid\! g_i\neq0\text{ for }i\leq l\}.\phantom{\Big|}
\]
Let $|G_i|=q_i$.
Recall from~\eqref{e-prodchar} that $\inner{\chi,g}=\prod_{i=1}^n\inner{\chi_i,g_i}$ for all $(\chi,g)\in\hat{\cG}\times\cG$.
Using the specific form of~$\chi$ and~$\chi''$ and the orthogonality relations~\eqref{e-GChar} we compute
\[
  \sum_{g\in P_n}\!\inner{\chi'',g}
  =\prod_{i=1}^l(q_i-1)\!\!\prod_{\genfrac{}{}{0pt}{2}{i=l+1}{i\neq j}}^n \!\!q_i\sum_{g_j\in G_j}\!\!\inner{\chi''_j,g_j}
  \ \text{ and }\
  \sum_{g\in P_n}\!\inner{\chi',g}
  =\prod_{i=l+1}^n q_i\big(\prod_{i=1}^l \sum_{g_i\in G_i\backslash\{0\}}\inner{\chi_i',g_i}\big).
\]
Again by~\eqref{e-GChar} the first sum is zero, whereas the second one is not.
This shows that $\chi'\not\hspace*{-.6em}\widesim[2]_{\widehat{\cP_{\bP}}}\chi''$, contradicting the assumption
$\widehat{\cP_{\bP}}=\cP_{\bbP}$.
Hence~$\bP$ is hierarchical.

Let now $\bP=\bH(n;n_1,\ldots,n_t)$ be as in Remark~\ref{R-HierPoset} and let
$\cG=\cG_1\times\ldots\times\cG_t$ be as in~\eqref{e-Gfactors}.
Assume $|G_{s,1}|\neq|G_{s,2}|$ for some~$s$.
For $a=1,2$ put $\chi_a:=(\veps,\ldots,\veps,\chi^{(a)}_s,\veps,\ldots,\veps)\in\hat{\cG_1}\times\ldots\times\hat{\cG_t}$, where
$\chi^{(1)}_s:=(\chi_{s,1},\veps,\ldots,\veps),\,\chi^{(2)}_s:=(\veps,\chi_{s,2},\veps,\ldots,\veps)\in\hat{\cG_s}$ and $\chi_{s,1},\chi_{s,2}$ are non-principal characters.
Due to~\eqref{e-HPPweight} the characters $\chi_1,\,\chi_2$ have the same dual poset weight, and thus
$\chi_1\widesim[2]_{\cP_{\bbP}}\chi_2$.

Put $m=\sum_{i=1}^{s}n_i$.
Then $P_m=\{(g_1,\ldots,g_s,0,\ldots,0)\mid g_i\in\cG_i,\,\wt_{\text{H}}(g_s)=n_s\}$, see~\eqref{e-HPweight}.
Using $S:=\prod_{i=1}^{s-1}|\cG_i|$ and~\eqref{e-GChar} we compute
\[
  \sum_{g\in P_m}\inner{\chi_1,g}
    =S\hspace*{-.8em}\sum_{ \genfrac{}{}{0pt}{2}{g_s\in\cG_s}{\wt_{\text{H}}(g_s)=n_s}}\hspace*{-.8em}\inner{\chi_s^{(1)},g_s}
      =S\prod_{j=2}^{n_s}(|G_{s,j}|-1)\hspace*{-.8em}\sum_{g_{s,j}\in G_{s,j}\backslash\{0\}}\hspace*{-.8em}\inner{\chi_{s,1},g_{s,1}}
      =-S\prod_{j=2}^{n_s}(|G_{s,j}|-1).
\]
In the same way we obtain $\sum_{g\in P_m}\inner{\chi_2,g}=-S(|G_{s,1}|-1)\prod_{j=3}^{n_s}(|G_{s,j}|-1)$.
Now $|G_{s,1}|\neq|G_{s,2}|$ implies that $ \sum_{g\in P_m}\inner{\chi_1,g}\neq  \sum_{g\in P_m}\inner{\chi_2,g}$, and thus
$\chi_1\not\hspace*{-.6em}\widesim[2]_{\widehat{\cP_{\bP}}}\chi_2$.
Again, this contradicts our assumption.
Thus, $|G_{s,1}|=\ldots=|G_{s,n_s}|$ for all~$s$, and this concludes the proof.
\end{proof}


The converse  of the previous result is true as well.
\begin{theo}\label{T-DualPartPoset2}
Let $\bP=\bH(n;n_1,\ldots,n_t)$ and~$\cG$ be as in~\eqref{e-Gfactors}, where
 $|G_{i,1}|=\ldots=|G_{i,n_i}|$ for all~$i=1,\ldots,t$.
Then $\widehat{\cP_{\bP}}=\cP_{\bbP}$ for the induced partitions~$\cP_{\bP}$ and~$\cP_{\bbP}$  on~$\cG$
and~$\hat{\cG}$.
As a consequence,~$\cP_{\bP}$ is reflexive.
\end{theo}

\begin{proof}
First of all, the last statement follows from $|\cP_{\bbP}|=|\cP_{\bP}|=n+1$ with Theorem~\ref{P-LPplus}.
\\
Let $\cP_{\bP}=P_0=\{0\}\mmid P_1\mmid\ldots\mmid P_n$ and
$\cP_{\bbP}=\bar{P}_0=\{\veps\}\mid\bar{P}_1\mmid \ldots\mmid\bar{P}_n$.
Define $n_0:=0$ and
\[
    N_s:=\sum_{i=0}^{s-1} n_i\text{ for }s=1,\ldots,t+1\text{ and }N_0:=0.
\]
For each $\ell,m=1,\ldots,n$ there exists unique indices~$s_m\in\{1,\ldots,t\}$ and $r_\ell\in\{0,\ldots,t-1\}$ such that
\begin{equation}\label{e-ml}
   m=N_{s_m}+\mu_m,\text{ where }1\leq\mu_m\leq n_{s_m}\ \text{ and }
   \ell=n-N_{t-r_\ell+1}+\lambda_\ell,\text{ where }1\leq\lambda_\ell\leq n_{t-r_\ell}.
\end{equation}
Set $s_0=\mu_0=0$ and $r_0=\lambda_0=0$.
Then~\eqref{e-HPweight} and~\eqref{e-HPPweight} show that
\begin{align}
  &P_m=\{(g_1,\ldots,g_{s_m-1},g_{s_m},0,\ldots,0)\in\cG\mid g_i\in\cG_i\text{ for }i\leq s_m\text{ and }\wt_{\text{H}}(g_{s_m})=\mu_m\},   \label{e-Pm}\\[1ex]
  &\bar{P}_\ell=\{(\veps,\ldots,\veps,\chi_{t-r_\ell},\chi_{t-r_\ell+1},\ldots,\chi_t)\in\hat{\cG}\mid
   \chi_i\in\hat{\cG_i}\text{ for }i\geq t-r_\ell\text{ and }\wt_{\text{H}}(\chi_{t-r_\ell})=\lambda_\ell\}. \label{e-barPell}
\end{align}
Let $\chi\in\hat{\cG}$, say $\chi\in\bar{P}_\ell$. We show that for each~$m$, the sum $\sum_{g\in P_m}\inner{\chi,g}$
does not depend on the choice of~$\chi$, but just on the index~$\ell$.
It suffices to consider $m>0$ and $\ell>0$.
For ease of notation write $s:=s_m,\,\mu:=\mu_m,\,r:=r_\ell,$ and $\lambda:=\lambda_\ell$.
Furthermore, put $|G_{i,j}|=q_i$ for all $j=1,\ldots,n_i$. Hence $|\cG_i|=q_i^{n_i}$.
\\[.7ex]
1) Let $s<t-r$. Then $\sum_{g\in P_m}\inner{\chi,g}=|P_m|=\prod_{i=1}^{s-1}q_i^{n_i}(q_s-1)^{\mu}\binom{n_s}{\mu}$.
\\
2) For $s=t-r$ we compute
$\sum_{g\in P_m}\inner{\chi,g}=\prod_{i=1}^{s-1}q_i^{n_i}\sum_{\genfrac{}{}{0pt}{2}{g_s\in\cG_s}{\wt_{\text{H}}(g_s)=\mu}}\inner{\chi_s,g_s}
    =\prod_{i=1}^{s-1}q_i^{n_i}K_{\mu}^{(n_s,q_s)}(\lambda)$, where the latter is the classical Krawtchouk coefficient
from~\eqref{e-Krawtclassic}.
\\
3) Finally, let $s>t-r$.
Then
$\sum_{g\in P_m}\inner{\chi,g}=\sum_{(g_1,\ldots,g_{s-1})}
   \sum_{\genfrac{}{}{0pt}{2}{g_s\in\cG_s}{\wt_{\text{H}}(g_s)=\mu}}\inner{\chi,g}$,
where the first summation is over all $(g_1,\ldots,g_{s-1})\in\cG_1\times\ldots\times\cG_{s-1}$.
This equals
\[
    \Big(\sum_{(g_1,\ldots,g_{s-1})}
       \inner{(\chi_1,\ldots,\chi_{s-1}),(g_1,\ldots,g_{s-1})}\Big)
       \Big(\sum_{\genfrac{}{}{0pt}{2}{g_s\in\cG_s}{\wt_{\text{H}}(g_s)=\mu}}\inner{\chi_s,g_s}\Big)=0,
\]
where the last identity follows because the first factor is zero due to~\eqref{e-GChar} and $\chi_{t-r}\neq\veps$.

In all three cases $\sum_{g\in P_m}\inner{\chi,g}$ depends only on the index~$\ell$ of the block $\bar{P}_\ell$
containing~$\chi$.
As a consequence, $\cP_{\bbP}\leq\widehat{\cP_{\bP}}$, and hence
$|\cP_{\bbP}|\geq|\widehat{\cP_{\bP}}|\geq|\cP_{\bP}|=|\cP_{\bbP}|$.
Thus  $\cP_{\bbP}=\widehat{\cP_{\bP}}$, and this concludes the proof.
\end{proof}

In the proof we also computed the Krawtchouk coefficients.

\begin{cor}\label{C-HposetKraw}
Let~$\cP=\bH(n;n_1,\ldots,n_t)$ and~$\cG$ be as in~\eqref{e-Gfactors}.
Assume $|G_{i,j}|=q_i$ for all $i,j$.
With the notation as in~\eqref{e-ml} --~\eqref{e-barPell}, the Kratwchouk matrix of $(\cP_{\bP},\cP_{\bbP})$
is $K\in\C^{(n+1)\times(n+1)}$, where
\[
  K_{\ell,m}=\left\{\begin{array}{cl}
       \prod_{i=1}^{s_m-1}\hspace*{-.2ex}q_i^{n_i}\,(q_{s_m}-1)^{\mu_m}\binom{n_{s_m}}{\mu_m},&\text{ if }s_m<t-r_\ell,\\
       0,&\text{ if }s_m>t-r_\ell,\\
       \prod_{i=1}^{s_m-1}\hspace*{-.2ex}q_i^{n_i}\,K_{\mu_m}^{(n_{s_m},q_{s_m})}(\lambda_\ell),&\text{ if }s_m=t-r_\ell.
       \end{array}\right.
\]
Thus by~\eqref{e-MacWLinear}, any code $\cC\leq\cG$ satisfies
\begin{equation}\label{e-MacWlinearPoset}
     (B_0,\ldots,B_n)=\frac{1}{|\cC|}(A_0,\ldots,A_n)K,
\end{equation}
where $A_m=|\cC\cap P_m|$ and $B_m=|\cC^\perp\cap \bar{P}_m|$ for all~$m$.
\end{cor}

In the  case where $\cG=G^n$, i.e., $G_{i,j}=G$ for all~$i,\,j$, then, using $|G|=q$, the above results in
\begin{equation}\label{e-Krawtq}
  K_{\ell,m}=\left\{\begin{array}{cl}
       q^{m-\mu_m}(q-1)^{\mu_m}\binom{n_{s_m}}{\mu_m},&\text{ if }s_m<t-r_\ell,\\
       0,&\text{ if }r_\ell>t-s_m,\\
       q^{m-\mu_m} K_{\mu_m}^{(n_{s_m},q)}(\lambda_\ell),&\text{ if }s_m=t-r_\ell.
       \end{array}\right.
\end{equation}

These results cover several cases in the literature.
Let $\cG=\F^n$, which we identify with its character module with the aid of a generating character~$\chi$  as in
Theorem~\ref{T-RFrob}.
Thus, the dual code~$\cC^{\perp}$ of $\cC\subseteq\F^n$ is simply the orthogonal space $\cC^{\pperp}\subseteq\F^n$
as in that theorem, and we will simply write~$\cC^{\perp}$.
As we have discussed in Definition~\ref{D-DualPartchi} and Example~\ref{E-FPartNotIndep}, the dual $\wcP^{^{[\chi]}}$ of
a partition~$\cP$ does in general depend on~$\chi$.
It is easy to see, however, that this dependence does not materialize for partitions $\cP_{\bP}$ induced by a poset~$\bP$.
This is due to the fact that the poset weight is based on the Hamming weight so that all computations simply reduce to the orthogonality relations~\eqref{e-GChar}; see the previous proofs.
For hierarchical posets, the dual partition is induced by the dual poset, and thus obviously independent from the choice
of~$\chi$.
All of this ensures that our setting applies to the situations discussed in the literature, and which we will address next.

If $\cG=\F^n$ (thus $G_i=\F$ for all~$i$), the identity~\eqref{e-MacWlinearPoset} with the Krawtchouk coefficients
in~\eqref{e-Krawtq} corresponds to the identity presented by Kim and Oh~\cite[Thm.~4.4]{KiOh05}.
Specializing even further to the hierarchical poset $\bP=\bH(n;n)$ we simply obtain the classical Krawtchouk coefficients and the MacWilliams identity for the Hamming weight distribution.
On the other hand, specializing to the poset $\bP=\bH(n;1,\ldots,1)$ and using $1<\ldots<n$, see Example~\ref{E-Poset1}(b),
we obtain the MacWilliams identity for the Rosenbloom-Tsfasman weight.
The Krawtchouk coefficients  are
\begin{equation}\label{e-KrawChain}
     K_{\ell,m}=\left\{\begin{array}{cl}
              1,&\text{if }m=0,\\
         q^{m-1}(q-1),&\text{if }\ell<n+1-m\neq n+1,\\ -q^{m-1},&\text{if }\ell=n+1-m,\\ 0,&\text{if }\ell>n+1-m,\end{array}\right.
\end{equation}
which coincides with \cite[Thm.~3.1]{DoSk02} by Dougherty and Skriganov.
Pinheiro and Firer~\cite{PiFi12} generalized the results of~\cite{KiOh05} to poset block structures.
In our terminology this is the case where $\cG_i=\F^{k_i}$ for all~$i$, thus
$\cG=\F^{k_1}\times\ldots\times\F^{k_n}$, and our results agree with those in \cite[Thm.~1, Thm.~2]{PiFi12}.


\medskip
We close this section by extending the above result to matrix codes over finite commutative Frobenius rings endowed with a poset metric.
This will come as an immediate consequence of the MacWilliams identity for the induced product partition and the induced
symmetrized partition derived in Section~\ref{SS-IndPart}.
Again, if~$\bP$ is a chain, we obtain the MacWilliams identity for the Rosenbloom-Tsfasman
metric for matrix codes over Frobenius rings and thus recover a result obtained by
Dougherty and Skriganov~\cite{DoSk02} for the special case of matrix codes over fields.

Let~$R$ be a finite commutative Frobenius ring of order~$q$.
The free module $R^{s\times n}$ of all $(s\times n)$-matrices over~$R$ may be
identified with the product $R^{sn}=R^n\times\ldots\times R^n$
by mapping $M\in R^{s\times n}$ to the vector $(M_1,\ldots,M_s)\in R^{sn}$, where~$M_i$ is the $i$-th row of~$M$.
The standard bilinear form $(v,w)\mapsto v\cdot w$ on $R^{sn}$ thus reads as $(M,N)\mapsto \text{Tr}(MN\T)$
on $R^{s\times n}$, where $\text{Tr}$ denotes the trace.
As  in Section~\ref{SS-FrobRing} we identify~$R$ with~$\hat{R}$ via a fixed generating character~$\chi$.
Then the identification of~$R^{sn}$ with $\widehat{R^{sn}}$ as in Theorem~\ref{T-RFrob}(a) is given by the isomorphism
$M\mapsto \chi_M$, where $\chi_M(N)=\chi(\text{Tr}(MN\T))$.

Let $\bP=([n],\leq)$ be a hierarchical poset on~$[n]$.
The product partition enumerator of a code~$\cC\subseteq R^{s\times n}$ is given by
\[
    \pe_{(\cP_{\bP})^n,\cC}=\sum_{M\in\cC}\prod_{i=1}^s X_{i,\wt_{\bP}(M_i)}\in\C[X_{i,j}\mid i=1,\ldots,s,\,j=0,\ldots,n].
\]
The coefficient of a monomial $\prod_{i=1}^s X_{i,m_i}$ is thus the cardinality of all matrices in~$\cC$ whose $i$-th row has poset weight~$m_i$.
By Theorem~\ref{T-ProdPartDual}, $\widehat{(\cP_{\bP})^n}=\widehat{\cP_{\bP}}^n=(\cP_{\bbP})^n$.
Hence Theorem~\ref{T-weprodpart} and~\eqref{e-MacWProd} yield
\begin{equation}\label{e-MacWRTMatrix}
  \pe_{(\cP_{\bbP})^n,\cC^\perp}({\bf X}_i,i=1,\ldots,s)
  =\frac{1}{|\cC|}\pe_{(\cP_{\bP})^n,\cC}(K{\bf X}_i,i=1,\ldots,s),
\end{equation}
where ${\bf X}_i=(X_{i,0},\ldots,X_{i,n})\T$ and~$K$ is the Krawtchouk matrix of $(\cP_{\bbP},\,\cP_{\bP})$.

If~$\bP$ is the chain $1<\ldots< n$ and $R=\F_q$ is a field, the
weight~$\wt_{\bP}$ is the Rosenbloom-Tsfasman weight on~$\F_q^n$,  and the enumerator $\pe_{(\cP_{\bP})^n,\cC}$
has been coined the T-enumerator by Dougherty and Skriganov~\cite[Sec.~3]{DoSk02}.
The MacWilliams identity in~\eqref{e-MacWRTMatrix} with the Krawtchouk coefficients in~\eqref{e-KrawChain}
appears in~\cite[Thm.~3.1]{DoSk02}, where it has been derived by
direct technical computations tailored to the specific situation of this  weight and without the aid of character theory.

In the same way, Theorem~\ref{T-wesymmpart} yields a MacWilliams identity for the symmetrized partition $(\cP_{\bP})^s_{\text{sym}}$ on~$R^{s\times n}$ and its dual partition.
Again, for fields this has been derived in~\cite[Thm.~3.2]{DoSk02}, where the enumerator is called the H-enumerator.
Another direct computational proof has been presented by Trinker~\cite{Tr09}, where the composition vector and symmetrized partition enumerator of the Rosenbloom-Tsfasman weight are called the type distribution and the type polynomial,
respectively.\footnote{In~\cite{DoSk02} and~\cite{Tr09}, the authors use a reversed inner dot product.
This results in the partition~$\cP_{\bP}$ being self-dual, in our terminology, and no dual poset is needed.}

We conclude by mentioning that the \emph{cumulative Rosenbloom-Tsfasman weight}, defined as
$\wt_{\bP}(M):=\sum_{i=1}^s\wt_{\bP}(M_i)$ and introduced in~\cite{RoTs97}, does not satisfy a MacWilliams identity.
In~\cite{DoSk02}, the authors present a pair of codes~$\cC_1,\,\cC_2$ with the same
cumulative Rosenbloom-Tsfasman weight enumerator, but where the dual codes~$\cC_1^{\perp},\,\cC_1^{\perp}$
have different enumerators.

\appendix
\section{Proof of Theorem~\ref{T-SymmPartDual}}\label{S-App}
Let $\cP=P_1\mmid\ldots\mmid P_{M+1}$, where $P_{M+1}=\{0\}$.

Note that by definition of the induced symmetrized partition we have for any $g,g'\in G^n$
\begin{equation}\label{e-SymmProd}
  g'\widesim[2]_{\cPnsym}g\Longleftrightarrow \exists\;\tau\in S_n: \tau(g')\widesim[2]_{\cP^n}g,
\end{equation}
where $\tau(g'_1,\ldots,g'_n)=(g'_{\tau(1)},\ldots,g'_{\tau(n)})$ and~$S_n$ is the symmetric group on~$n$ symbols.

We show first $\widehat{\cP}_{\text{sym}}^{\,n}\leq\widehat{\cPnsym}$.
Let $\chi,\,\chi'\in\hat{G}^n$ such that $\chi\widesim[2]_{\widehat{\cP}_{\text{sym}}^{\,n}}\chi'$.
By~\eqref{e-SymmProd} there exists a permutation $\tau\in S_n$ such that
$\tau(\chi')\widesim[2]_{\wcP^n}\chi$.
Let~$Q$ be a block of~$\cPnsym$.
Then, by definition of the symmetrized partition, there exists
$(m_1,\ldots,m_n)\in\{1,\ldots,M+1\}^n$ and a subset $S'\subseteq S_n$ such that
$Q=\bigcup_{\sigma\in S'}\hspace*{-2.5em}\raisebox{.5ex}{$\cdot$}\hspace*{2.2em}P_{m_{\sigma(1)}}\times\ldots\times P_{m_{\sigma(n)}}$.
Using $\wcP^n=\widehat{\cP^n}$ from Theorem~\ref{T-ProdPartDual} we have
\[
  \sum_{g\in P_{m_{\sigma(1)}}\times\ldots\times P_{m_{\sigma(n)}}}\inner{\tau(\chi'),g}
  =\sum_{g\in P_{m_{\sigma(1)}}\times\ldots\times P_{m_{\sigma(n)}}}\inner{\chi,g}
\]
for all $\sigma\in S'$, and thus $\sum_{g\in Q}\inner{\tau(\chi'),g}=\sum_{g\in Q}\inner{\chi,g}$.
But,~$\tau^{-1}(Q)=Q$, and thus the first sum is $\sum_{g\in Q}\inner{\chi',g}$.
This shows $\chi\widesim[2]_{\widehat{\cPnsym}}\chi'$ and hence $\widehat{\cP}_{\text{sym}}^{\,n}\leq\widehat{\cPnsym}$.
\footnote{From this one can easily conclude equality if~$\cP$ is reflexive. Indeed, in that case $|\cP|=|\wcP|$ and thus
$|\cPnsym|=|\widehat{\cP}_{\text{sym}}^{\,n}|$.
By Remark~\ref{R-DualNega}(c) $\widehat{\cP}_{\text{sym}}^{\,n}\leq\widehat{\cPnsym}$ implies
$\cPnsym=\wwcP_{\text{sym}}^{\,n}\leq
\widehat{\phantom{\big|}\hspace*{1.8em}}\hspace*{-2.1em}\widehat{\cP}_{\text{sym}}^{\,n}
\leq
 \widehat{\phantom{\big|}\hspace*{1.8em}}\hspace*{-2.1em}\widehat{\cPnsym}\leq\cPnsym$,
and thus we have equality at each step, which in turn yields
$\widehat{\cPnsym}=\widehat{\cP}_{\text{sym}}^{\,n}$.}

For the converse, let $\chi\widesim[2]_{\widehat{\cPnsym}}\chi'$.
By the above it suffices to show that there exists some $\tau\in S_n$ such that $\tau(\chi')\widesim[2]_{\wcP^n}\chi$.
Put
\[
  A_{i,j}=\sum_{g\in P_j}\inner{\chi_i,g} \text{ and } B_{i,j}=\sum_{g\in P_j}\inner{\chi_i',g}
  \text{ for all } i=1,\ldots,n,\,j=1,\ldots,M+1.
\]
We will show that there exists a permutation $\tau\in S_n$ such that
\begin{equation}\label{e-ABij}
  A_{i,j}=B_{\tau(i),j}\text{ for all } i,\,j,
\end{equation}
because this implies $\chi'_{\tau(i)}\widesim_{\wcP}\chi_i$ for all~$i$, thus
$\tau(\chi')\widesim[2]_{\widehat{\cP}^{n}}\chi$ and consequently $\chi'\widesim[2]_{\widehat{\cP}_{\text{sym}}^{\,n}}\chi$.

Note that $A_{i,M+1}=B_{i,M+1}=1$.
Furthermore, for all $(j_1,\ldots,j_n)\in\{1,\ldots,M+1\}^n$
\[
  \sum_{g\in P_{j_1}\times\ldots\times P_{j_n}}\inner{\chi,g}
  =\sum_{g\in P_{j_1}\times\ldots\times P_{j_n}}\prod_{i=1}^n\inner{\chi_i,g_i}
  =\prod_{i=1}^n\sum_{g\in P_{j_i}}\inner{\chi_i,g}=\prod_{i=1}^n A_{i,j_i}.
\]
Making use of $P_{M+1}=\{0\}$, the blocks of~$\cPnsym$ can be described as follows, see~\eqref{e-CompQ}.
Define the set $\cS:=\{(s_1,\ldots,s_M)\in\N_0^M\mid |s|\leq n\}$, where $|s|:=\sum_{m=1}^M s_m$.
For each $s\in\cS$, let
\begin{equation}\label{e-tvector}
  t:=t^{(s)}:=(\underbrace{1,\ldots,1}_{s_1},\underbrace{2,\ldots,2}_{s_2},\ldots,\underbrace{M,\ldots,M}_{s_M},\underbrace{M+1,\ldots,M+1}_{n-|s|})
  \in\{1,\ldots,M+1\}^n.
\end{equation}
Then $\cPnsym=(Q_s)_{s\in\cS}$, where
\[
  Q_s:=\bigcup_{\sigma\in S_n}P_{t_{\sigma(1)}}\times\ldots\times P_{t_{\sigma(n)}}.
\]
Note that this union is in general not disjoint.
But every subset $P_{t_{\sigma(1)}}\times\ldots\times P_{t_{\sigma(n)}}$ in this union
appears with the same multiplicity.
It is given by the cardinality of the stabilizer subgroup of~$t=t^{(s)}$ in~$S_n$ and denoted by~$f_s$.
Thus we have
\[
  \sum_{g\in Q_s}\inner{\chi,g}=f_s^{-1}\sum_{\sigma\in S_n}\prod_{i=1}^n\sum_{g\in P_{t_{\sigma(i)}}}\inner{\chi_i,g}
  =f_s^{-1}\sum_{\sigma\in S_n}\prod_{i=1}^n A_{\sigma(i),t_i}.
\]
Now our assumption $\chi\widesim[2]_{\widehat{\cPnsym}}\chi'$ along with $A_{i,M+1}=B_{i,M+1}=1$ for all~$i$ implies
\begin{equation}\label{e-chiQs}
  \sum_{\sigma\in S_n}\prod_{i=1}^{|s|} A_{\sigma(i),t_i}=\sum_{\sigma\in S_n}\prod_{i=1}^{|s|} B_{\sigma(i),t_i}
  \text{ for all }
  s\in\cS \text{ and where~$t$ is as in~\eqref{e-tvector}}.
\end{equation}
Note that $t_i\leq M$ for all~$t_i$ appearing in~\eqref{e-chiQs}.

The crucial step in order to establish~\eqref{e-ABij} is the fact that the expressions in these identities are
evaluations of the elementary multi-symmetric polynomials.
In order to make this precise, we consider the polynomial ring $R:=\Q[X_{i,j}\mid i=1,\ldots,n,\,j=1,\ldots,M]$ in $nM$ independent
indeterminates.
The symmetric group~$S_n$ acts on~$R$ via $\sigma(X_{i,j})=X_{\sigma(i),j}$ for all $i,j$.
It is a classical result, see~\cite{Noe15} or~\cite[Thm.~1]{Bri04}, that the invariant ring~$R^{S_n}$ under this group
action is generated by the elementary multi-symmetric polynomials, that is, $R^{S_n}:=\Q[e_s\mid s\in\cS]$, where
\[
   e_s=\sum_{\sigma\in S_n}\prod_{i=1}^{|s|}X_{\sigma(i),t_i},\text{ where~$t$ is as in~\eqref{e-tvector}.}
\]
More important for our purposes, however, is the fact that, similar to the elementary symmetric polynomials, the elementary
multi-symmetric polynomials appear as coefficients of a certain polynomial, see for instance~\cite[(1.1)]{Vac05}.
Namely, in the polynomial ring  $R[T_1,\ldots,T_M]$ with~$M$ independent indeterminates we have
\[
   \sum_{s\in\cS}T^se_s=\prod_{i=1}^n\Big(1+\sum_{j=1}^MT_jX_{i,j}\Big),
\]
where $T^s:=T_1^{s_1}\cdot\ldots\cdot T_M^{s_M}$.
Now~\eqref{e-chiQs} yields
\[
  f:=\prod_{i=1}^n\Big(1+\sum_{j=1}^MT_jA_{i,j}\Big)=\prod_{i=1}^n\Big(1+\sum_{j=1}^MT_jB_{i,j}\Big) \text{ in }
  \C[T_1,\ldots,T_M].
\]
In other words, we have two factorizations of the polynomial $f\in\C[T_1,\ldots,T_M]$ into linear, thus prime, factors.
Since all factors are normalized with constant term equal to~$1$, they must coincide up to ordering.
All of this shows that there exists a permutation~$\tau\in S_n$ satisfying~\eqref{e-ABij},
and this concludes the proof.\hfill{\large $\Box$}

\section*{Acknowledgment}
I would like to thank Marcus Greferath and Navin Kashyap for very inspiring suggestions concerning
this research project.
A major part of the final write-up took place during a research stay at the University of Z\"urich, and I
am grateful to Joachim Rosenthal and his research group for the generous hospitality.

\bibliographystyle{abbrv}
\bibliography{literatureAK,literatureLZ}

\end{document}